%% file: WSDM21.tex
\documentclass[sigconf]{acmart}

\usepackage{tikz}
\usetikzlibrary{shapes,matrix,arrows,decorations.pathreplacing,automata,positioning,decorations.pathmorphing}

\usepackage{diagbox}
\usepackage{multirow}

\usepackage[font=Large]{subfig}
\AtBeginDocument{%
  \providecommand\BibTeX{{%
    \normalfont B\kern-0.5em{\scshape i\kern-0.25em b}\kern-0.8em\TeX}}}

\input{preamble}

\setcopyright{acmcopyright}
\copyrightyear{2021}
\acmYear{2021}
\acmDOI{10.1145/1122445.1122456}
\acmConference[WSDM '21]{WSDM '21: Web Search and Data Mining}{March 08--12, 2021}{Jerusalem, Israel}
\acmBooktitle{WSDM '21: Web Search and Data Mining,
  March 08--12, 2021, Jerusalem, Israel}
\acmPrice{15.00}
\acmISBN{978-1-4503-XXXX-X/18/06}

\begin{document}

\title{Balance Maximization in Signed Networks via Edge Deletions}


%\iffalse
\author{Kartik Sharma}
\affiliation{%
  \institution{IIT Delhi}
 %\country{France}
}
\email{Kartik.Sharma.cs117@cse.iitd.ac.in }

\author{Iqra Altaf Gillani}
\affiliation{%
  \institution{IIT Delhi}
  }
\email{iqraaltaf@cse.iitd.ac.in }

\author{Sourav Medya}
\affiliation{%
  \institution{Northwestern University}
}
\email{sourav.medya@kellogg.northwestern.edu}

\author{Sayan Ranu}
\affiliation{%
 \institution{IIT Delhi}}
 \email{sayanranu@cse.iitd.ac.in}
\author{Amitabha Bagchi}
\affiliation{%
 \institution{IIT Delhi}}
 \email{bagchi@cse.iitd.ac.in}

%\fi

\renewcommand{\shortauthors}{Sharma et al.}

\input{0_abstract}

%\keywords{Signed graphs, balance maximization, network design, combinatorial optimization}

\maketitle

%\maketitle
\input{1_introduction}
%\input{2_previous_work}
\input{3_problem}

\input{4_method_spectral}
\input{5_methods_NON}
\input{6_experiments}

\input{7_conclusion}
\input{7_1_paper_appendix}

%\input{7_appendix}

\iffalse
\begin{acks}
To Robert, for the bagels and explaining CMYK and color spaces.
\end{acks}
\fi
%\clearpage
\bibliographystyle{ACM-Reference-Format}
\bibliography{WSDM21}

%\clearpage
\input{appendix}

\end{document}

%% file: preamble.tex
\usepackage{soul}
\usepackage{url}
\usepackage{epstopdf}
\usepackage{float}
\usepackage[utf8]{inputenc}
\usepackage{graphicx}
\usepackage{amsmath}
\usepackage{booktabs}
\usepackage[noend]{algorithmic}
\usepackage{algorithm}

\usepackage{booktabs} % For formal tables
\usepackage{graphicx,array}
\usepackage{balance}
\usepackage{color,soul,xspace}

\usepackage{comment}
\usepackage{epsfig}
\usepackage{subfig}
\usepackage{mathrsfs,mdwlist,enumitem,amsthm,bm}
\DeclareMathOperator*{\argmax}{arg\,max}
\setlist{nolistsep,leftmargin=*}
\def\greedy{\textsc{Greedy}\xspace}
\def\rg{\textsc{Rg}\xspace}

\newtheorem{thm}{\textbf{Theorem}}

\newtheorem{defn}{\textbf{Definition}}
\newtheorem{lem}{\textbf{Lemma}}

\newtheorem{observation}{\textbf{Observation}}

\newtheorem{example}{\textbf{Example}}
\newtheorem{prob}{\textbf{Problem}}

\def\RR{\mathbb{R}}

%% file: 0_abstract.tex
\begin{abstract}
     In \emph{signed} networks, each edge is labeled as either positive or negative. The edge sign captures the polarity of a relationship. \emph{Balance} of signed networks is a well-studied property in graph theory. In a balanced (sub)graph, the vertices can be partitioned into two subsets with negative edges present only across the partitions. Balanced portions of a graph have been shown to increase coherence among its members and lead to better performance. While existing works have focused primarily on finding the largest balanced subgraph inside a graph, we study the \emph{network design} problem of maximizing balance of a target community (subgraph). %More specifically, we study the problem of maximizing the balance via edge deletions. 
     In particular, given a budget $b$ and a community of interest within the signed network, we aim to make the community as close to being balanced as possible by deleting up to $b$ edges. Besides establishing \emph{NP-hardness}, we also show that the problem is non-monotone and non-submodular.  To overcome these computational challenges, we propose heuristics based on the spectral relation of balance with the \emph{Laplacian spectrum} of the network. Since the spectral approach lacks approximation guarantees, we further design a \emph{greedy} algorithm, and its \emph{randomized} version, with provable bounds on the approximation quality. The bounds are derived by exploiting \emph{pseudo-submodularity} of the balance maximization function. Empirical evaluation on eight real-world signed networks establishes that the proposed algorithms are effective, efficient, and scalable to graphs with millions of edges.
\end{abstract}
%have been studied in the literature with special focus %%%on their balance property. A balanced network is one whose vertices can be partitioned into two subsets with negative edges only across the partitions. Such graphs can model the existence of polarised communities in online networks. While several works have focused on finding balanced subgraphs, optimizing balance through network design has not gained much attention. In this paper, we study the problem of maximizing the balance via edge deletion. In particular, we are given a target community which we need to get close to being balanced by removing a given budget of edges. We show that the problem is NP-hard and propose a heuristic based on the relation of balance in signed graphs and the \emph{Laplacian spectrum}. We also show that the given problem is \emph{pseudo-submodular}, which allows us to design a randomized greedy algorithm with provable approximation guarantees. Empirical evaluation on real-world signed graphs establishes that our proposed methods are efficient, effective, and scalable to graphs with millions of edges.

%% file: 1_introduction.tex
\vspace{-0.10in}
\section{Introduction and Related Work}
%% balance and its applications
Graphs can model various complex systems such as knowledge graphs~\cite{kg}, road networks~\cite{medya2018noticeable}, communication networks~\cite{myinfocom}, and social networks~\cite{kempe}. Typically, nodes represent entities, and edges characterize relationships between pairs of entities.  %While graphs can capture weight or direction of  edges between the nodes, in its typical form, they do not capture the nature of the edge. For example, in an acquaintance network or a social network where the interaction between the nodes (representing people) is captured using the edges between them. In such cases, the edges do not reflect whether there is friendship or dislike, alliance or conflict between the individuals. 
\emph{Signed} graphs further enhance the representative power of graphs by capturing the \emph{polarity} of a relationship through  \emph{positive} and \emph{negative} edge labels~\cite{harary1953notion, huffner2007optimal, ordozgoiti2020finding}. For example, if a graph represents social interactions, a positive edge would denote friendly interaction, and a negative edge would indicate a hostile relationship. Similarly, in a collaboration network, positive edges may indicate complementary skill sets, whereas negative edges would indicate disparate skills.

Signed graphs were first studied by Harary et al.~\cite{harary1953notion} with particular focus on their \emph{balance}. A balanced signed graph is one in which the vertices can be partitioned into two sets such that all edges inside each partition have a positive sign and all the negative signed edges are across the partitions. Balance is correlated with both positive and negative side-effects on a community. On the positive side,  balanced communities are positively correlated with performance in financial networks where edges represent trading links~\cite{omid, figueiredo2014maximum}. On the negative side, in social networks,  \emph{balanced communities} % polarized communities arising due to difference in opinions, personal biases, political affiliations and preferences~\cite{garimella2017long}. In particular, balance 
often promote ``echo-chambers'', reduce diversity of opinions, and ultimately lead to more polarized viewpoints~\cite{garimella2017long}. 
%% why do we want to increase balance -- examples
%When a community is balanced, there is coherence in  the taste of the community members, which leads to better performance and synchronized activity~\cite{omid,garimella2018political}. For example, 

Owing to the correlation of balance with several higher-order functional traits, it is natural to measure how far a community is from being balanced. For example, in financial networks, it is important to evaluate how the community may be engineered to further improve its balance. On the other hand, in social networks, an adversary, such as a political party, may be interested in polarizing the community in its favor by further increasing its balance. To avoid such adversarial attacks, it is important to know the weak links in a community so that they can be safeguarded. 

In this paper, we address these applications by studying the problem of \emph{maximizing balance via edge deletions (\textsc{Mbed})}.
In the \textsc{Mbed} problem, we are given a graph, a target community within this graph, and a budget $b$. Our goal is to remove $b$ edges, such that the community gets as \emph{close} to being balanced as possible. We formally define the notion of balance closeness in \S~\ref{sec:formulation}. Deleting an edge would correspond to actions such as \emph{unfollowing} or \emph{blocking} a connection. If increasing balance is desirable, then \textsc{Mbed} provides a mechanism towards achieving the goal. On the other hand, \textsc{Mbed} also measures how susceptible a community is to adversarial attacks by revealing how much the balance can be increased through a small number of deletions, and which are these critical edges that must be protected.%The \emph{distance} of a community $H$ from being balanced is defined as $V(H)\setminus V(S(H))$, where $V(H)$ is the vertex set of $H$ and $V(S(H))$ is the vertex set of the largest balanced subgraph of $H$.  In the real world, deleting an edge would correspond to convincing two entities to stop interacting, so that the overall community becomes more stable. %More specifically, the balance of a community is defined as the size of the largest induced subgraph within the community that is balanced.

%Naturally, balanced subgraphs %Consider the particular situation, where we require to identify these 
%represent communities that are stable \cite{belardo2014balancedness}.
%and create \emph{echo chambers} over the social media \cite{garimella2018political}. For this, we require to remove the noise links, i.e., the edges going across the polarized communities, to increase the size of chambers (respective communities).
% some more examples

\vspace{-0.15in}
\subsection{Related Work}
\label{sec:relatedwork}
The problem we study falls in the class of \emph{network design} problems. In  network design, the goal is to modify the network so that an objective function modeling a desirable property is optimized.
Examples of such objective functions include optimizing shortest path distances (traffic and sustainability improvement)~\cite{meyerson2009,dilkina2011,lin2015,medya2018noticeable}, increasing centrality of target nodes by adding a small set of edges~\cite{crescenzi2015,ishakian2012framework,medya2018group}, optimizing the $k$-core\cite{medyakcore,zhou2019k}, manipulating node similarities \cite{dey2019manipulating}, and boosting/containing influence on social networks~\cite{kimura2008minimizing,chaoji2012recommendations,medya_influence}. %[[I am trying to emphasize that both network design and balance are important research threads, but no work exists in their intersection. Since we have 1 entire page for reference, it may be safe to cite some of the network design papers. I have shortened it a little by removing the coreness citations. Do we need a long list of unrelated network design problems here? AB]] %However, in this paper, we address a network design problem that differs from the described ones in terms of the objective function, i.e., maximizing balance in networks. 
%% introduce the problem
%Signed graphs have been studied in the literature for a long time since the pioneering work by Harary \cite{harary1953notion}, who was particularly interested in the notion of balance. %The concept of balance was however first introduced by Heider \cite{heider1946attitudes} in his psychological study of sentiments. Later, Cartwright and Harary \cite{cartwright1956structural} extended Heider's concept of balance to the signed graphs. Most earlier works focused on testing the balance property in signed graphs. Examples include the pseudo-boolean programming formulation by Hammer \cite{hammer1977pseudo}, the linear-time algorithm of Harary and Kabell \cite{harary_kabell1980simple}, the minimum number of sign changes approach by Akiyama et al. \cite{akiyama1981balancing}.
% Problem variations for balance in signed graphs

While several works exist on finding balanced subgraphs~\cite{harary1953notion,figueiredo2014maximum,dasgupta2007algorithmic,huffner2007optimal, ordozgoiti2020finding}, work on optimizing balance through network design is rather limited. The only work is by Akiyama et al. \cite{akiyama1981balancing}, where they study the minimum number of sign flips needed to make a graph balanced. However our work is different for several reasons. %There are several differences between this and our approach. 
 First, \cite{akiyama1981balancing} does not have any notion of a budget constraint. Second, the cascading impact of a sign flip and an edge deletion on the balance of a graph is significantly different. Third,~\cite{akiyama1981balancing} lacks evaluation on large real world graphs containing millions of edges. Finally, from a practicality viewpoint, selectively flipping the sign of an edge is difficult since the edge sign encodes the nature of interaction between the two entities (endpoints) of the edge. %For instance, it represents how the two entities feel about each other in a social network interaction graph. It is unclear what such a flipping of signs would require in this setting. 
 In contrast, deleting an edge is a more lightweight task as it only involves stopping further interactions with a chosen node. %In the case of social platforms, such a cessation of interaction can be achieved using \emph{unfollowing}, \emph{blocking} or \emph{deleting} a connected entity.% functionality that is typically available in such platforms.

Several studies related to identifying large balanced subgraphs exist. Poljak and Turz\'{i}k addressed the problem of finding a maximum weight balanced subgraph and showed an equivalence with max-cut in a graph with a general weight function~\cite{poljak1986polynomial}. Other approaches include finding balanced subgraphs with the maximum number of vertices~\cite{figueiredo2014maximum,ordozgoiti2020finding} and edges~\cite{dasgupta2007algorithmic} in the context of biological networks. H\"{u}ffner et al.~\cite{huffner2007optimal} gave an exact algorithm for finding such balanced subgraphs using the idea of graph separators. %While the above papers mostly focused on finding the existing balanced subgraph, in this paper, we focus on the problem of maximizing balance by deleting a few edges in the graph. 
More recently, Ordozgoiti et al. \cite{ordozgoiti2020finding} studied the problem to identify the maximum balanced subgraph in a given graph and designed efficient and effective heuristics.%As already mentioned, while the above papers mostly focused on finding the existing balanced subgraph, in this paper, we focus on the problem of maximizing balance by deleting a few edges in the graph.
\vspace{-0.05in}
\subsection{Contributions}
%Keeping in mind the fact that the balance of a subgraph is correlated with several desirable properties, 
%In this paper, we study the  optimizing balance. 
Our key contributions are summarized as follows.
%In this paper, The key contr%More specifically, we aim to maximize balance by deleting a given number of edges inside a given graph or a part of the graph (subgraph). 
%% methods and contributions
%We first prove that this problem is NP-hard by showing a reduction from Set Union Knapsack problem. Then, we present two efficient methods to solve this problem. The first approach is based on the relation of balance in signed graphs with that of the spectrum of corresponding Laplacian matrix \cite{hou2003laplacian}. The second approach uses the submodular optimization to solve the \textsc{Mbed} problem. In particular, we show that the objective function of the problem is pseudo submodular and give an efficient greedy and randomised greedy algorithm for the problem.
% results from experiments
\begin{itemize}
    \item We propose the novel network design problem of \emph{\underline{m}aximizing \underline{b}alance in a target subgraph via \underline{e}dge \underline{d}eletion (\textsc{Mbed})}. We establish that \textsc{Mbed} is NP-hard, non-submodular and non-monotonic (\S~\ref{sec:formulation}).
    \item Since NP-hardness makes an optimal algorithm infeasible, we propose an efficient, algebraically-grounded heuristic that exploits the connection of balance in a signed graph with the \emph{spectrum} of its {\em Laplacian matrix} (\S~\ref{sec:methods_spectral}). Although this spectral approach is extremely efficient, it lacks an approximation guarantee. We overcome this weakness by establishing that \textsc{Mbed} is \emph{pseudo-submodular}, which is then utilized to design  \emph{greedy} algorithms with provable quality guarantees  (\S~\ref{sec:methods_submodular}). 
    \item We extensively benchmark the proposed methodologies on an array of eight real-world signed graphs. Our experiments establish that the proposed methodologies are effective, efficient, and scalable to million-sized graphs (\S~\ref{sec:expts}). % Need to add other Experimental results and analysis
\end{itemize}

%% file: 3_problem.tex
%\clearpage
\vspace{-0.05in}
\section{Problem Definition}
\label{sec:formulation}
In this section, we introduce the concepts central to our problem. All important notations used in our work are summarized in Table~\ref{tab:symbols}.%The idea is to remove a few edges from the network to increase the amount of pre-existing balance. The definitions of a signed graph and its balance are as follows:
\vspace{-0.05in}
\begin{defn}[Signed graph]
A {\em signed graph}, $\Gamma = (G, \sigma)$ is a undirected graph $G = (V,E)$ along with a mapping $\sigma : E \rightarrow \{-1, +1\}$, called its {\em edge labelling}, that assigns a sign to each edge. 
\end{defn}
\vspace{-0.05in}

Given a signed graph $\Gamma=((V,E),\sigma)$, we use the notation $E^+ = \{ e \in E : \sigma(e) = +1\}$ and $E^- = \{ e \in E : \sigma(e) = -1\}$ to denote the set of positive and negative edges in $\Gamma$ respectively. %Given a walk $W$ we say the sign of the walk is $\sgn(W) = \prod_{e \in W} \sigma(e)$. 

\vspace{-0.05in}
\begin{defn}[Balanced graph]
A signed graph $\Gamma = ((V,E),\sigma)$ is said to be balanced if there exists a partition $(V_1, V_2)$ of $V$ such that for every $(u,v)\in E$ with $\sigma(u,v) = -1$, $u \in V_1$ iff $v \in V_2$. 
\end{defn}
\vspace{-0.05in}

\begin{figure}[t]
    \centering
    %\vspace{-0.40in}
    \resizebox*{1.04\linewidth}{!}{       
        \input{example_fig}
    }
    \vspace{-0.20in}
    \caption{This figure shows a series of signed graphs. We use the following coloring scheme. The balanced subgraphs contain the colored nodes in blue (marked in `o') and red (in $`\times'$) representing node partition sets $V_1$ and $V_2$. Nodes outside the balanced component are in white. The current balance of (a) is $\Delta(\Gamma)=6$, whereas in (b)-(d) a single edge deletion increases $\Delta(\Gamma)$ to $8$. (e) Illustration of why \textsc{Mbed} is not submodular.}
    \label{fig:cycle_example}
    \vspace{-0.05in}
\end{figure}
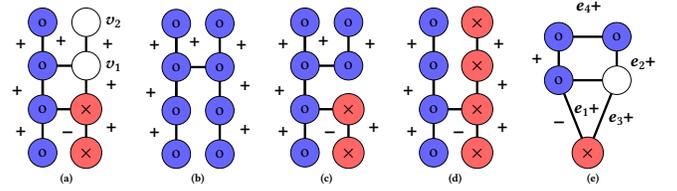

\vspace{-0.05in}
\begin{example}
Consider the signed graph in (Fig.~\ref{fig:cycle_example}(a)). The subgraph induced by the coloured nodes is a balanced subgraph since it can be partitioned into disjoint sets $V_1$ (blue) and $V_2$ (red) with positive edges within the partitions and negative edges across partitions.
\end{example}
\vspace{-0.05in}

%[[I am finding the caption of Figure 1 very confusing. I understand what is going on but in the first readign it will lead to confusion. AB.]]

\vspace{-0.05in}

\begin{defn}[Current Balance ($\Delta(\Gamma)$) \cite{figueiredo2014maximum}] Given a signed graph $\Gamma$, the  current balance $\Delta(\Gamma)$ is the maximum number of nodes in any  induced subgraph that is connected and balanced. The largest connected induced balanced subgraph is denoted by $S(\Gamma)$, and thus, $\Delta(\Gamma) = |V(S(\Gamma))|$.
\end{defn}
\vspace{-0.05in}

It is worth noting that the largest connected induced balanced subgraph  might not be unique. 
%%%%%%%%%%%%%%%%%%%%%%%%%%%%%%%%%%% Edge DELETION %%%
%\subsection{Edge Deletion}

We solve a network design problem where the balance is maximized via edge deletions. The modified graph is denoted as $\Gamma_X$ after the deletion operation of edge set $X$ on $\Gamma$. Deletion of an edge (positive or negative) may increase the balance of a graph.

\vspace{-0.05in}
\begin{example}
The current balance of the graph in Fig.~\ref{fig:cycle_example}(a) is 6. Deleting any negative or positive edge increases the balance to 8 (Fig.~\ref{fig:cycle_example}(b)-(d)). Note that deleting an edge may initiate a cascading impact and bring in multiple nodes into the balanced subgraph.
%\textbf{No consider. Add a figure. The same figure should also be enough for example 1.}
\end{example}
\vspace{-0.05in}

%An application of increasing balance would be to increase the uniformity inside a community. Thus to define the optimizing balance problem, we would like a connected component that is balanced in our resultant balanced subgraph. 

% Next, we define the problem of balance maximization via edge deletion. 

\vspace{-0.05in}
\begin{prob}[Maximizing Balance via Edge Deletion (\textsc{Mbed})] Given a signed (sub) graph $H$, a candidate edge set $\mathbb{C}$ and a budget $b$, find the set, $B \subset \mathbb{C}$ of $b$ edges to be deleted such that $f(B)= \Delta(H_B) - \Delta(H)$, i.e., the number of nodes in $S(H_B)$, is maximized. Here, $H_B = (V(H), E(H) \setminus B)$ is the subgraph of $H$ formed by deleting the edge set $B$ from $H$.
\end{prob}
\vspace{-0.05in}

%\textbf{What's $H^m$? Where has it been defined?It is defined two paragraphs earlier}

Note that maximizing $f(B)=\Delta(H_B)-\Delta(H)$ is equivalent to maximizing $\Delta(H_B)$. We envision $H$ to be the target community where we would like to maximize balance. $\mathbb{C}$ denotes the edges that may be deleted, which  may be the entire edge set of $H$.  

%\textbf{Practical implications of the input set:} 

%The original problem of finding maximum balance is NP-hard \cite{figueiredo2014maximum}. We show that maximizing balance via edge deletion is also NP-hard. 
\vspace{-0.05in}
\subsection{Problem Characterization}
\label{sec:character}

\begin{thm}\label{thm:nphard_MBED}
The \textsc{Mbed} problem is NP-hard.
\end{thm} 
\vspace{-0.05in}
\textsc{Proof.} We reduce \textsc{Mbed} from the Set Union Knapsack Problem 
 \cite{goldschmidt1994note}. The details are in Section ~\ref{app:nphard}.
 \vspace{-0.05in}
\begin{lem}
\label{lem:nonmonotonic}
The optimization function $f(B)$ of \textsc{Mbed} is non-monotonic, i.e., an edge deletion may lead to a decrease in current balance.
\end{lem}
\vspace{-0.05in}
\vspace{-0.05in}
\begin{proof} Consider the path $a-b-c-d$ with only edge $(b,c)$ being negative. The current balance is 4 since the entire graph is balanced. If we delete any edge, the balance decreases to at most $3$.
\end{proof}
\vspace{-0.05in}
\begin{table}[t]
\vspace{-0.10in}
\centering
{\scriptsize
\begin{tabular}{cl}
\toprule
\textbf{Symbol} &\textbf{Definition and Description} \\
\midrule
$\Gamma=((V,E),\sigma)$  & Signed undirected graph with sign fn. $\sigma$ \\
$S(\Gamma)$ & Largest balanced (connected and induced) subgraph of $\Gamma$ \\
$\Delta(\Gamma)$ &  $|V(S(\Gamma))|$  \\
%$\Gamma_X$ & $\Gamma$ after edge deletion operation of edge set $X$\\
$\mathbb{C}$ & Candidate edge set\\
$b$ & Budget (i.e., $\#$edges to be deleted) \\
$L(\Gamma)$ & Laplacian matrix of signed graph $\Gamma$\\
$\lambda_1(\Gamma)$ & Smallest eigenvalue of $L(\Gamma)$\\
$\bm{u,v}$ & Vectors (bold lower case)\\
$v_i$ & $i^{th}$ entry of $\bm{v}$\\
$H_X$ & Subgraph $H$ after deleting edges $X \subseteq \mathbb{C}$\\
%$f_X(Y)$ &$f(X\cup Y)-f(X)$\\
$cep(H, x)$ & Set of contradictory edge-pairs for subgraph\\
 & $H$ with one end at node $x$\\
\bottomrule
\end{tabular}}
\caption{Frequently used symbols}
\label{tab:symbols}
\vspace{-0.2in}
%\vspace{-0.05in}
\end{table}
A function $f(.)$ is \emph{submodular} \cite{kempe} if the marginal gain by adding an element $e$ to a subset $S$ is equal or higher than the same in a superset $T$. Mathematically, it satisfies: 
\begin{equation}
	f(S\cup \{e\})-f(S)\geq f(T \cup \{e\})-f(T)
\end{equation}
for all elements $e$ and all pairs of sets $S \subseteq T$ and $e\notin S, e \notin T$.

\vspace{-0.05in}
\begin{lem}
\label{lem:submodular}
$f(B)$ is not sub-modular \footnote{We show a stronger result that it is not even proportionally-submodular in Sec. \ref{app:propsubmodular}.}. 
\end{lem}
\vspace{-0.05in}
\vspace{-0.05in}
\begin{proof}
    %Consider a balanced subgraph of $H$, $S(H)$ has a partition $V_1$ and $V_2$. A node $v$ is outside $S(H)$ and it is connected to $V_1$ with positive edges $e_1$ and $e_2$, $V_2$ with another positive edge $e_3$. Thus the node $v$ cannot be the part of $S(H)$. Consider an edge $e_4$ inside $V_1$ which can be removed without making the graph disconnected. 
   In Fig.~\ref{fig:cycle_example}e, let $S = \{ e_4\}, T = \{e_1, e_4\}$. %Then, $f(\{e_1\}) = 0$ and $f(\{e_1, e_4\}) = 0$.%, since even after removing any of these edges it is not possible to add the node $v$ to $S(H)$. Note that 
    Here, $f(S \cup \{e_2\}) = 0, f(T\cup \{e_2\}) = 1$. Thus, $f(S \cup \{e_2\})- f(S) < f(T \cup \{e_2\})- f(T)$.
\end{proof}
\vspace{-0.05in}
 %For submodular (and monotone) functions, the iterative greedy algorithm of adding the element with the maximum marginal gain gives an approximation within a factor of $(1-e^{-1})$ \cite{nemhauser1978}. However, we show that our problem is pseduo-submodular and give an approximation bound based on it.
%\iffalse

%\fi
Owing to NP-hardness, devising an optimal algorithm for \textsc{Mbed} is not feasible in polynomial time. Furthermore, due to the optimization function being non-monotonic and non-submodular, greedy algorithms exploiting these properties are also not applicable. We overcome these computational challenges through two different approaches: a \emph{spectral} approach built on \emph{signed graph Laplacians} (\S~\ref{sec:methods_spectral}) and approximation schemes based on \emph{pseudo-submodular} optimization (\S~\ref{sec:methods_submodular}). %design efficient approaches to solve the \textsc{Mbed} problem in Sections \ref{sec:methods_spectral} and \ref{sec:methods_submodular} where the former one describes spectral methods and the latter one discusses 

%%%%%%%%%%%%%%%%%%%%%%%%%%%%%%%%

%% file: example_fig.tex
\subfloat[]{
\begin{tikzpicture}[>=latex]
  \Huge
  \tikzstyle{every node} = [circle,draw=black]
  \foreach \x in {0,1}
  \foreach \y in {0,1} 
  \node  (\x\y) at (1.5*\x,1.5*\y) {};
  \node  (12) at (1.5,3) {\textcolor{white}{o}};
  \node  (02) at (0,3) {};
  \node (015) at (0,-1.5) {};
  \node (1515) at (1.5,-1.5) {};
 
  % Lines
  \draw[line width=0.8mm](00) -- (01);
  \draw[line width=0.8mm](00) -- (10);
  \draw[line width=0.8mm](01) -- (11);
  \draw[line width=0.8mm](11) -- (10);
  \draw[line width=0.8mm](01) -- (02);
  \draw[line width=0.8mm](11) -- (12);
  \draw[line width=0.8mm](02) -- (015);
  \draw[line width=0.8mm](10) -- (1515);
  
  % Add signs on edges
  \node at (10) [draw=none,above right=20 pt] {$\bm{+}$};
  \node at (1515) [draw=none,above right=20 pt] {$\bm{+}$};
  \node at (11) [draw=none,above left=20 pt] {$\bm{+}$};
  \node at (01) [draw=none,below left=20 pt] {$\bm{+}$};
  \node at (00) [draw=none,below right=20 pt] {$\bm{-}$};
  \node at (02) [draw=none,below left=15 pt] {$\bm{+}$};
  \node at (015) [draw=none,above left=15 pt] {$\bm{+}$};
  \node at (12) [draw=none,below right=15 pt] {$\bm{+}$};
  
  %Add vertex ids
  \node at (11) [draw=none,right=10 pt] {$\bm{v_1}$};
  \node at (12) [draw=none,right=10 pt] {$\bm{v_2}$};
  
  % Fill nodes with signs and colors
  \node at (00) [fill=blue!60] {o};
  \node at (11) [fill=white!60] {\textcolor{white}{o}};
  \node at (01) [fill=blue!60] {o};
  \node at (10) [fill=red!60] {$\times$};
  \node at (1515) [fill=red!60] {$\times$};
  \node at (02) [fill=blue!60] {o};
  \node at (015) [fill=blue!60] {o};
    \end{tikzpicture}
}
\subfloat[]{
\begin{tikzpicture}[>=latex]
  \Huge
  \tikzstyle{every node} = [circle,draw=black]
  \foreach \x in {0,1}
  \foreach \y in {0,1} 
  \node  (\x\y) at (1.5*\x,1.5*\y) {};
  \node  (12) at (1.5,3) {};
  \node  (02) at (0,3) {};
  \node (015) at (0,-1.5) {};
  \node (1515) at (1.5,-1.5) {};
  
  % Diagonal lines
  \draw[line width=0.8mm](00) -- (01);
  %\draw[line width=0.8mm](00) -- (10);
  \draw[line width=0.8mm](01) -- (11);
  \draw[line width=0.8mm](11) -- (10);
  \draw[line width=0.8mm](01) -- (02);
  \draw[line width=0.8mm](11) -- (12);
  \draw[line width=0.8mm](02) -- (015);
  \draw[line width=0.8mm](10) -- (1515);
  
  % Add signs on edges
  \node at (10) [draw=none,above right=20 pt] {$\bm{+}$};
  \node at (1515) [draw=none,above right=20 pt] {$\bm{+}$};
  \node at (11) [draw=none,above left=20 pt] {$\bm{+}$};
  \node at (01) [draw=none,below left=20 pt] {$\bm{+}$};
  \node at (00) [draw=none,below right=20 pt] {\textcolor{white}{$\bm{-}$}};
  \node at (02) [draw=none,below left=15 pt] {$\bm{+}$};
   \node at (015) [draw=none,above left=15 pt] {$\bm{+}$};
  \node at (12) [draw=none,below right=15 pt] {$\bm{+}$};

  % Fill nodes with signs and colors
  \node at (00) [fill=blue!60] {o};
  \node at (01) [fill=blue!60] {o};
  \node at (11) [fill=blue!60] {o};
  \node at (10) [fill=blue!60] {o};
  \node at (02) [fill=blue!60] {o};
  \node at (12) [fill=blue!60] {o};
  \node at (015) [fill=blue!60] {o};
  \node at (1515) [fill=blue!60] {o};
    \end{tikzpicture}
}
% \subfloat[]{
% \begin{tikzpicture}[>=latex]
%   \Huge
%   \tikzstyle{every node} = [circle,draw=black]
%   \foreach \x in {0,1}
%   \foreach \y in {0,1} 
%   \node  (\x\y) at (1.5*\x,1.5*\y) {};

%   % Diagonal lines
%   \draw[line width=0.8mm](00) -- (01);
%   \draw[line width=0.8mm](00) -- (10);
%   \draw[line width=0.8mm](01) -- (11);
%   %\draw[line width=0.8mm](11) -- (10);
  
%   % Add signs on edges
%   \node at (10) [draw=none,above right=20 pt] {\textcolor{white}{$\bm{+}$}};
%   \node at (11) [draw=none,above left=20 pt] {$\bm{+}$};
%   \node at (01) [draw=none,below left=20 pt] {$\bm{+}$};
%   \node at (00) [draw=none,below right=20 pt] {$\bm{-}$};
  
%   % Fill nodes with signs and colors
%   \node at (00) [fill=blue!60] {o};
%   \node at (01) [fill=blue!60] {o};
%   \node at (11) [fill=blue!60] {o};
%   \node at (10) [fill=red!60] {$\times$};
  
%     \end{tikzpicture}
% }
\subfloat[]{
\begin{tikzpicture}[>=latex]
  \Huge
  \tikzstyle{every node} = [circle,draw=black]
  \foreach \x in {0,1}
  \foreach \y in {0,1} 
  \node  (\x\y) at (1.5*\x,1.5*\y) {};
  \node  (12) at (1.5,3) {};
  \node  (02) at (0,3) {};
  \node (015) at (0,-1.5) {};
  \node (1515) at (1.5,-1.5) {};
  
  % Diagonal lines
  \draw[line width=0.8mm](00) -- (01);
  \draw[line width=0.8mm](00) -- (10);
  \draw[line width=0.8mm](01) -- (11);
  %\draw[line width=0.8mm](11) -- (10);
  \draw[line width=0.8mm](01) -- (02);
  \draw[line width=0.8mm](11) -- (12);
  \draw[line width=0.8mm](02) -- (015);
  \draw[line width=0.8mm](10) -- (1515);
  
  % Add signs on edges
  \node at (10) [draw=none,above right=20 pt] {\textcolor{white}{$\bm{+}$}};
  \node at (1515) [draw=none,above right=20 pt] {$\bm{+}$};
  \node at (11) [draw=none,above left=20 pt] {$\bm{+}$};
  \node at (01) [draw=none,below left=20 pt] {$\bm{+}$};
  \node at (00) [draw=none,below right=20 pt] {$\bm{-}$};
  \node at (02) [draw=none,below left=15 pt] {$\bm{+}$};
  \node at (015) [draw=none,above left=15 pt] {$\bm{+}$};
  \node at (12) [draw=none,below right=15 pt] {$\bm{+}$};
  
  % Fill nodes with signs and colors
  \node at (00) [fill=blue!60] {o};
  \node at (01) [fill=blue!60] {o};
  \node at (11) [fill=blue!60] {o};
  \node at (10) [fill=red!60] {$\times$};
  \node at (02) [fill=blue!60] {o};
  \node at (12) [fill=blue!60] {o};
  \node at (1515) [fill=red!60] {$\times$};
  \node at (015) [fill=blue!60] {o};
    \end{tikzpicture}
}
\subfloat[]{
\begin{tikzpicture}[>=latex]
  \Huge
  \tikzstyle{every node} = [circle,draw=black]
  \foreach \x in {0,1}
  \foreach \y in {0,1} 
  \node  (\x\y) at (1.5*\x,1.5*\y) {};
  \node  (12) at (1.5,3) {};
  \node  (02) at (0,3) {};
  \node (015) at (0,-1.5) {};
  \node (1515) at (1.5,-1.5) {};
  
  % Diagonal lines
  \draw[line width=0.8mm](00) -- (01);
  \draw[line width=0.8mm](00) -- (10);
  %\draw[line width=0.8mm](01) -- (11);
  \draw[line width=0.8mm](11) -- (10);
  \draw[line width=0.8mm](01) -- (02);
  \draw[line width=0.8mm](11) -- (12);
  \draw[line width=0.8mm](02) -- (015);
  \draw[line width=0.8mm](10) -- (1515);
  
  % Add signs on edges
  \node at (10) [draw=none,above right=20 pt] {$\bm{+}$};
  \node at (1515) [draw=none,above right=20 pt] {$\bm{+}$};
  \node at (11) [draw=none,above left=20 pt] {\textcolor{white}{$\bm{+}$}};
  \node at (01) [draw=none,below left=20 pt] {$\bm{+}$};
  \node at (00) [draw=none,below right=20 pt] {$\bm{-}$};
  \node at (02) [draw=none,below left=15 pt] {$\bm{+}$};
  \node at (015) [draw=none,above left=15 pt] {$\bm{+}$};
  \node at (12) [draw=none,below right=15 pt] {$\bm{+}$};
  
  % Fill nodes with signs and colors
  \node at (00) [fill=blue!60] {o};
  \node at (01) [fill=blue!60] {o};
  \node at (02) [fill=blue!60] {o};
  \node at (11) [fill=red!60] {$\times$};
  \node at (10) [fill=red!60] {$\times$};
  \node at (12) [fill=red!60] {$\times$};
  \node at (1515) [fill=red!60] {$\times$};
  \node at (015) [fill=blue!60] {o};
\end{tikzpicture}
}
\subfloat[]{
\begin{tikzpicture}[>=latex]
  \Huge
  \tikzstyle{every node} = [circle,draw=black]
    \node  (12) at (2,3) {};
    \node  (02) at (0,3) {};
    \node  (01) at (0,1.5) {};
    \node  (11) at (2,1.5) {};
    \node  (00) at (1,-1) {};
  
   % Diagonal lines
  \draw[line width=0.8mm](02) -- (12);
  \draw[line width=0.8mm](02) -- (01);
  \draw[line width=0.8mm](12) -- (11);
  \draw[line width=0.8mm](01) -- (11);
  \draw[line width=0.8mm](01) -- (00);
  \draw[line width=0.8mm](11) -- (00);
  
   Add signs on edges
  \node at (02) [draw=none,above right=20 pt] {$\bm{e_4+}$};
  \node at (12) [draw=none,below right=15 pt] {$\bm{e_2+}$};
  \node at (01) [draw=none,above left=15 pt] {$\bm{+}$};
  \node at (11) [draw=none,below left=20 pt] {$\bm{e_1+}$};
  \node at (00) [draw=none,above right=25 pt] {$\bm{e_3+}$};
 \node at (00) [draw=none,above left=25 pt] {$\bm{-}$};
  
  % Fill nodes with signs and colors
  \node at (02) [fill=blue!60] {o};
  \node at (12) [fill=blue!60] {o};
  \node at (01) [fill=blue!60] {o};
  \node at (11) [fill=white!60] {\textcolor{white}{o}};
  \node at (00) [fill=red!60] {$\times$};
\end{tikzpicture}
}
% \subfloat[]{
% \label{fig:submodular}
% \includegraphics[width=1.4in]{pdf/submodular.pdf}
% }

%% file: 4_method_spectral.tex
\vspace{-0.05in}
\section{The Spectral Approach}
\label{sec:methods_spectral}
%In this section, we present spectral approaches to solve the MBED problem. In particular, we attempt to minimize the smallest eigenvalue of the Laplacian of the given graph. We begin with some notations and definitions. 

%[[To save space, this section should be abbreviated in my opinion. AB]]

Given a signed graph $\Gamma=((V,E),\sigma)$, let $A$ be its adjacency matrix where $A_{ij} = \sigma(i,j)$ for $(i,j)\in E$, and $0$ otherwise. Furthermore, let $D$ be the diagonal degree matrix defined as $D_{ii} = d(i)$, where $d(i)$ is the vertex degree, i.e., the total number of edges incident on vertex $i$. We define the corresponding signed Laplacian as follows.

\vspace{-0.05in}
\begin{defn}[Signed Laplacian]
The Laplacian of a signed graph $\Gamma = ((V,E), \sigma)$, denoted as $L(\Gamma)$ is a symmetric matrix $|V| \times |V|$ matrix defined as $L(\Gamma)=D(\Gamma)-A(\Gamma)$, i.e., $L_{ii} = d_i$, and  $L_{ij} = -\sigma(i,j)$ if $(i,j) \in E$ and 0 otherwise for $i\ne j$. 
\end{defn}
\vspace{-0.05in}

%For solving the MBED problem, we minimize the smallest eigenvalue of the Laplacian of the given initial signed subgraph. The approach is inspired from a similar spectral approach used in \cite{ordozgoiti2020finding} where the objective is to find a balanced subgraph by removing nodes iteratively. These spectral methods are motivated from the following interesting relation between the balance of the graph and the value of the minimum eigenvalue of its Laplacian \cite{hou2003laplacian}.
\vspace{-0.05in}
\begin{lem}[\cite{hou2003laplacian}]
\label{lem:eigenvalue_bal}
Given a signed graph $\Gamma=((V,E),\sigma)$, % with corresponding Laplacian $L(\Gamma)$. Then, 
$\Gamma$ is balanced iff the smallest eigenvalue of the Laplacian $\lambda_1(\Gamma)=0$.
\end{lem}
\vspace{-0.05in}

It has been further shown that $\lambda_1(\Gamma)$ is a measure of how "far" the graph is from being balanced~\cite{li2009note, belardo2014balancedness}. 

% \begin{lem}[\cite{li2009note}]
% \label{lem:min_max_eigen}
% For a signed graph $\Gamma=((V,E),\sigma)$ with corresponding Laplacian $L$.
% \begin{equation*}
%     \lambda_1 (\Gamma) \le \min{\{ \lambda_n(\Gamma^{'}): \Gamma - \Gamma^{'} \text{is balanced}\}}
% \end{equation*}
% where $\lambda_1$ and $\lambda_n$ denote the smallest and the largest eigenvalues of the Laplacian respectively. 
% \end{lem}
\vspace{-0.05in}
\begin{lem}[\cite{belardo2014balancedness}]
\label{lem:frust_no_ind}
Given a signed graph $\Gamma=((V,E),\sigma)$ with $\lambda_1(\Gamma)$ as the smallest eigenvalue of the corresponding Laplacian.
\begin{equation*}
    \lambda_1 (\Gamma) \le \nu(\Gamma) \le \epsilon(\Gamma)
\end{equation*}
where $\nu(\Gamma)$ ($\epsilon(\Gamma)$) denotes the \textit{frustration number} (\textit{frustration index}), i.e., the minimum number of vertices (edges) to be deleted such that the signed graph is balanced. 
\end{lem}
\vspace{-0.05in}

% Using the lem \ref{lem:belardo}, we can show that  $ |V| - \lambda_1(\Gamma) \geq |V| - \nu(\Gamma)$. 
Note that $ \Delta(\Gamma) = |V| - \nu(\Gamma) $. Through Lemma \ref{lem:frust_no_ind}, for any given subgraph $H$, we have:
\vspace{-0.05in}
\begin{equation} \label{eq:balance_eigen}
    \Delta(H) = |V(H)| - \nu(H)\le |V(H)| - \lambda_1(H)
\end{equation}

\vspace{-0.10in}
\subsection{An Upperbound Based Algorithm}\label{sec:spectral_top}
%In this paper, our objective is to maximize $\Delta(H)$ via a given budget number of edge deletions. 
Since directly maximizing $\Delta(H)$ is NP-hard, we turn our focus to the upperbound provided by Eq.~\eqref{eq:balance_eigen}. It is evident that maximizing the upper bound is equivalent to  minimizing $\lambda_1(H)$. %We use $H_X$ to denote the subgraph $H$ after deleting edges $X \subseteq \mathbb{C}$. 
To minimize $\lambda_1(H)$, we first derive the following upper bound. 
\vspace{-0.05in}

\begin{lem}
\label{lem:edge_set_removal}
Given a signed graph $\Gamma$, a subgraph $H$, a candidate edge set $\mathbb{C}$, for a set $X \subseteq \mathbb{C}$, we have
\begin{equation}
    \lambda_1(H_X) \le \lambda_1(H) - \sum_{(i, j) \in X} {(v_i - \sigma(i, j)v_j)^2}
\end{equation}
$\bm{v}$ denotes the unit eigenvector of Laplacian $L(H)$ corresponding to the minimum eigenvalue $\lambda_1(H)$ and $v_i$ denotes the $i^{th}$ entry of $\bm{v}$. Recall, $H_X$ denotes the subgraph formed due to removal of edge set $X$ from $H$.
\end{lem}
\vspace{-0.05in}
\vspace{-0.05in}
\begin{proof}
Given a signed graph $\Gamma$ with $L(\Gamma)$ being its corresponding Laplacian. We know for any $\bm{u} \in \RR^{|V|}$, 
\begin{equation}
\label{eq:gen_qform}
\bm{u}^T L(\Gamma) \bm{u} = \sum_{(i,j) \in E^+} (u_i - u_j)^2  + \sum_{(i,j) \in E^-} (u_i + u_j)^2.
\end{equation}
Now, using Eq. \eqref{eq:gen_qform} for $L(H_X)$ and (unit) eigenvector $\bm{v}$ of $L(H)$ corresponding to $\lambda_1(H)$, we get 

\vspace{-0.15in}
    \begin{align*}
        \bm{v}^T L(H_X) \bm{v} &= \sum_{(i, j) \in E(H_X)} {(v_i - \sigma(i, j)v_j)^2} \\
        &= \sum_{(i, j) \in E(H)} {(v_i - \sigma(i, j)v_j)^2} - \sum_{(i, j) \in X} {(v_i - \sigma(i, j)v_j)^2} \\
        &= \bm{v}^T L(H) \bm{v} - \sum_{(i, j) \in X} {(v_i - \sigma(i, j)v_j)^2}.
    \end{align*}
 \vspace{-0.10in}
   
    Note that as $\lambda_1(H_X) = \min_{\bm{z}} {\frac{\bm{z}^T L(H_X) \bm{z}}{\bm{z}^T \bm{z}}}$, $\lambda_1(H_X) \le \frac{\bm{v}^T L(H_X) \bm{v}}{\bm{v}^T \bm{v}}$. Substituting $\frac{\bm{v}^T L(H) \bm{v}}{\bm{v}^T \bm{v}} = \lambda_1(H)$ and $\bm{v}^T \bm{v} = 1$, the result is proved. 
\end{proof}
\vspace{-0.10in}

 We denote the upper bound as the function $g$, where $g$ is
 
 \vspace{-0.10in}
    \begin{equation*}
        g(X) = \lambda_1(H) - \sum_{(i, j) \in X} {(v_i - \sigma(i, j)v_j)^2}
    \end{equation*}  %there exists a polynomial-time greedy algorithm to compute it optimally. %So, we'd try to minimize the minimum eigenvalue by minimizing the upper bound on it (as proved above). 
\vspace{-0.05in}
    
The upper bound $g(X)$ is easier to optimize than minimizing $\lambda_1(H)$. In particular, $g(X)$ is a modular function and hence greedily choosing the top-$b$ edges will achieve an optimal solution \cite{nemhauser1978}.
\vspace{-0.05in}

\begin{lem} \label{lemma:modular}
    $g(X)$ is modular (submodular and supermodular). 
\end{lem}
\vspace{-0.16in}

\begin{proof} 
The proof is in Section ~\ref{app:modular}.
\end{proof}
\textbf{Algorithm:} Since $g(X)$ is modular, we simply compute $g(\{e\})$, for each edge  $e=(i,j) \in \mathbb{C}$ and select the top-$b$ edges based on the value of $(v_i - \sigma(i, j)v_j)^2$, where $b$ is the budget.

% \paragraph{Time complexity:} The algorithm involved in this approach requires to compute the smallest \emph{eigenpair} of $L(H)$ only once. So, we can use the standard Lanczos method \cite{orecchia2012approximating,coakley2013fast}. This will give a time complexity of $O(md_{avg}|V(H)|)$ where $m$ is the number of iterations for which the algorithm is run, $d_{avg}$ is the average number of nonzero elements in a row of the corresponding Laplacian matrix and is equal to the average vertex degree in subgraph $H$ plus one (for the Laplacian matrix's diagonal element).
The algorithm involved in this approach requires to compute the smallest \emph{eigenpair} of $L(H)$ only once. So, we can use the Locally Optimal Block Preconditioned Conjugate Gradient (LOBPCG) method proposed by Knyazev \cite{knyazev2001toward}. This method has theoretical guarantee on linear convergence, and the costs per iteration and the memory use are competitive with those of the Lanczos method \footnote{Lanczos algorithm \cite{orecchia2012approximating} (with Fast Multipole method \cite{coakley2013fast}) has a time complexity of $O(d_{avg}|V(H)|k)$ where $d_{avg}$ is the average number of nonzero elements in a row of the matrix and $k$ is the number of iterations of the algorithm.}.
\vspace{-0.05in}
\subsection{ Perturbation \& Iterative Algorithm}
We extend the described upper bound in Lemma \ref{lem:edge_set_removal} into a tighter expression and design another way to solve MBED in an iterative fashion. Similarly, the main idea is to compute change in the smallest eigenvalue $\lambda_1(H)$ of the Laplacian with a single edge deletion. We drop $H$ and use $\lambda_1(H)=\lambda_1$ where the context is understood. 

Let $\hat{\lambda}_1$ be the (exact) smallest eigenvalue of $\hat{L}(H)$, where $\hat{L}(H)$ is the perturbed version of $L(H)$ obtained by deleting a single edge $(i,j)\in E(H)$. Let $\delta=\lambda_2-\lambda_1$ be the eigengap of $L(H)$. For graphs that have sufficiently large eigengaps, we show the following result.

\vspace{-0.05in}
\begin{lem}\label{lemma:edge_perturb}
%Given a signed graph $\Gamma=(V,E,\sigma)$ with a subgraph $H$, 
 Given $\lambda_1$ is the smallest eigenvalue of $L(H)$ and $\bm{v}$ is the corresponding unit eigenvector, for $\delta\geq 4$ we have $\hat{\lambda}_1=\lambda_1-\left(v_i-\sigma(i,j)v_j\right)^2+O(1)$. 
\end{lem}
\vspace{-0.10in}
\begin{proof}
See App. \ref{sec:proof_edge_perturb}
\end{proof}
\vspace{-0.15in}
%\textbf{Note:} If we remove set of edges $S$ such that $|S|=k$, the relation can be generalised to $\hat{\lambda}_1(H)=\lambda_1(H)+\sum_{\forall k=i,j\in S}D_{kk}u(k)v(k)+\sum_{\forall i,j\in S}\sigma(i,j)u(i)v(j)+$
% Similar to the algorithm of \cite{Chen-TKDD:2016} we can have the following algorithm for our problem. 
\subsubsection{Algorithm: }We use Lemma \ref{lemma:edge_perturb} to design an iterative algorithm (Alg. \ref{alg:spectral_edge}). Given $\bm{v}$ as unit eigenvector corresponding to the smallest eigenvalue $\lambda_1$, we define score of an edge $e=(i,j)\in E(H)$ as $ (v_i-\sigma(i,j)v_j)^2$. We use this score to subsequently find the best edge from the candidate edge set $\mathbb{C}$ (lines $4-6$).
 In subsequent iterations (lines $2-8$) of the algorithm, we recompute the eigenpair (line $3$) corresponding to the minimum eigenvalue of the perturbed matrix after the deletion of the best edge (line $7$) and use LOBPCG method for all such iterations to achieve faster convergence.
%Since the smallest eigenpair needs to be re-computed in each iteration of Alg.~\ref{alg:spectral_edge}, we use LOBPCG method for all such iterations to achieve faster convergence.
\vspace{-0.05in}
\subsubsection{Limitations:}
 Alg.~\ref{alg:spectral_edge} does not provide any approximation guarantee and does not directly optimize the objective in MBED. Rather, it minimizes the smallest eigenvalue. Although it is known that in a balanced graph, $\lambda_1=0$, no result is known on the \emph{gradients} of change in balance with that of change in $\lambda_1$, i.e., the relationship between $\Delta(H_X)-\Delta(H)$ with $\lambda_1(H)-\lambda_1(H_X)$. %Naturally, there are no approximation guarantees of these approaches. In this section, we build approximation algorithms that relies on pseduo-submodular optimization. In \S~\ref{sec:expts}, we show these methods are more effective in practice as well.
%\end{itemize}
To address these weaknesses, we next directly optimize the objective function and show that MBED is \emph{pseudo-submodular}, which in turn allows us to provide an approximation guarantee on quality.
 \begin{algorithm}[t]
\caption{Spectral Edge Deletion}
\label{alg:spectral_edge}
{\scriptsize
\begin{algorithmic}[1] 
  \REQUIRE  The initial subgraph $H$, budget $b$, candidate set $\mathbb{C}$
\ENSURE A set $B$ of $b$ edges
  \STATE $H_0\leftarrow H$, $B\leftarrow \Phi$
   \FOR{$k=1$ to $b$}
   \STATE Compute eigenpair $\lambda_1(H_{k-1})$, $\bm{v}$
   \FOR{$e=(i,j)\in \mathbb{C}\setminus{B}$}
   \STATE Compute $score(e)= (v_i-\sigma(i,j)v_j)^2$
   \ENDFOR
    \STATE $e_k= \text{argmax}_{e\in \mathbb{C}}\text{ } score(e)$ 
    \STATE $B \leftarrow B \cup\{e_k\}$, $E(H_{k})=E(H_{k-1})\setminus{e_k}$
  \ENDFOR
 \RETURN $B$
  \end{algorithmic}}
  \end{algorithm}

%% file: 5_methods_NON.tex
\vspace{-0.10in}
\section{Approximation Algorithms}
\label{sec:methods_submodular}
%In this section, we present an efficient method based on submodular optimization to solve the MBED problem. 

%\end{example}

%\textbf{Assumption:} For the rest of the paper, we assume that the set $X$ is chosen such that $H_X$ and $H$ have same number of connected components. 

%[[Currently this section is quite hard to follow and is likely to irritate a reviewer. I suggest we put a roadmap of what we're going to do right in the beginning. Also if Thm 2 is our goal then we should clarify how the various lemmas and arguments help develop it. I think a top-down approach to presenting this material will make it more accessible for our reviewers. AB.]]

%In this section we design approximation algorithms using monotonic and pseduo-submodular properties of the objective function. 
In \S~\ref{sec:character}, we showed that MBED is not monotonic. We next show that if the set of deleted edges $X$ is selected strategically, then monotonicity can be guaranteed. If the optimization function is \emph{monotonic} and \emph{pseudo-submodular}, then greedy algorithms can produce approximation bounds. The rest of the section builds towards this result.
\vspace{-0.05in}
\begin{observation} \label{obs:monotonic}
    If the set of deleted edges $X$ is chosen such that $H_X$ and $H$ have same number of connected components, then the objective function $f(\cdot)$ is monotonic, i.e., $f(S \cup \{e\}) \geq f(S)$ $\forall S, e$.
\end{observation}
\vspace{-0.10in}
\begin{proof}
For all of the subsequent discussions, we will use $S(H)$ to denote the largest balanced subgraph of $H$ with the two vertex sets being $V_1$ and $V_2$.
  The deleted edge $X=\{e\}$ can fall in one of three categories. (1) both end points lie in $V_1$ (or equivalently $V_2$), in which case $\Delta(H)=\Delta(H_X)$ since $H_X$ and $H$ have same number of connected components. (2) One endpoint lies in $V_1$ and the other in $V_2$. Even in this case $\Delta(H)=\Delta(H_X)$. (3) One endpoint in $V_1$ (or $V_2$) and the other in $V(H)\setminus \{V_1\cup V_2\}$. In this case, the node in $V_1$ continues to stay there while the other endpoint may move into $V_1$ or $V_2$ and thus $\Delta(H)\leq\Delta(H_X)$.
\end{proof}
\vspace{-0.05in}
%[[It's not clear from this proof why the fact that removal that maintains connectedness is needed. Somewhere this should be mentioned. AB.]]

Choosing $X$ is in our control. Hence, we may assume that MBED is monotonic by ensuring that $X$ satisfies the constraint outlined in Obs.~\ref{obs:monotonic}. We next establish that although MBED is not submodular (Lem.~\ref{lem:submodular}), it is \emph{pseudo-submodular} (Thm. \ref{thm:local_pseudo_submodular}).% We use Lemmas \ref{lem:singleEdgeDeletion} and \ref{lemma:induction} to prove the main result.

%Under this assumption, it is easy to prove that objective function $f$ is monotonic, i.e., $f(S \cup {e}) \geq f(S)$ for all such $S, e$. 
%This shows that edge deletions cannot decrease the balance. In our proposed methods, we use $H_X$ to denote the subgraph $H$ after deleting edges $X \subseteq \mathbb{C}$. 

%%%%%%%%%%%%%%%%%%%%%%%%%%%
%\input{method_spectral}

%%%%%%%%%%%%%%%%%%%%%%%%%%%%%%

%We first show that the objective function $f$ in the MBED problem is not submodular
 %We begin with a few definitions and observations. 

\vspace{-0.05in}
\subsection{Pseudo-Submodularity}
We first prove that our objective function is pseudo-submodular (Thm. \ref{thm:local_pseudo_submodular}) and then provide approximations (Thms. \ref{thm:approx_RG} and \ref{thm:approx_greedy}) via \textit{Randomized Greedy} and \textit{Greedy} algorithms. 
\label{sec:pseudosubmodular}
%This section introduces the approximation ratio by showing that the objective function $f$  is pseudo submodular (Theorem \ref{thm:local_pseudo_submodular}).

%Let us consider the subgraph $H$ with its initial largest balanced connected subgraph as $S(H)$ and a partition $(V_1,V_2)$.
\vspace{-0.05in}
\begin{defn}[Contradictory Edge-pair]
    Given a subgraph $H$ with largest balanced subgraph $S(H)$ having balance partition $(V_1,V_2)$ two edges $e_1, e_2$ %that are connected to a node $x$ outside $S(H)$ 
    form a contradictory edge-pair if any of these conditions follow for some $u,u'\in V_1$ and $w,w'\in V_2$, and $x\notin V_1\cup V_2$:
    \begin{enumerate}
        \item $e_1 = (x, u)$ and $ e_2 = (x, w)$ such that
        $\sigma((x, u)) = \sigma((x, w))$.
        \item $e_1 = (x, u)$ and $ e_2 = (x, u')$ such that $\sigma((x,u)) = - \sigma((x, u'))$
        \item $e_1 = (x, w)$ and $ e_2 = (x, w')$ such that $\sigma((x,w)) = - \sigma((x, w'))$
    \end{enumerate}
\end{defn}
\vspace{-0.05in}

    We use $cep(H, x)$ to denote the set of contradictory edge-pairs for subgraph $H$ with one end at node $x$. A contradictory edge pair restricts node $x$ from contributing to the balance. This property is more formally expressed as follows.
    %Moreover, we say an edge to be a \textit{contradictory connection} if it constitutes a \textit{contradictory edge-pair}.
\vspace{-0.05in}
\begin{observation} \label{obs:nodetypes}
    A node $x$ will not be part of $S(H)$ if one of the following conditions hold: (1) $|cep(H, x)| > 0$, (2) the node $x$ is connected to $S(H)$ only via paths ending at a node $y$ where $|cep(H, y)| > 0$. %that has a contradictory edge-pair connected to it. 
    %\begin{enumerate}
        %\item $|cep(H, x)| > 0$ 
        % There exists at least one \textit{contradictory edge-pair} associated with the node $x$. 
       % \item The node $x$ is connected to $S(H)$ only via paths ending at a node that has a contradictory edge-pair connected to it. 
        %\item The node is not connected to $S(H)$. 
    %\end{enumerate}
\end{observation}
\vspace{-0.05in}

\vspace{-0.05in}
\begin{example}
In Fig.~\ref{fig:cycle_example}(a), nodes $v_1$ and $v_2$ are not part of the balanced subgraph $S(\Gamma)$ due to condition (1) and condition (2) respectively.
\end{example}
\vspace{-0.05in}

Obs.~\ref{obs:nodetypes} allows us to formally define when an edge deletion increases the balance.

\vspace{-0.05in}
\begin{observation}\label{obs:singleEdgeDeletion}
$f(\{e\}) > 0$ iff $(e, e') \in cep(H,x)$ for some $e' \in E$, $x \in V(H)$, and $|cep(H_{\{e\}}, x)| = 0$, i.e., following deletion of $e$, $x$ does not associate with any contradictory edge pair.
\end{observation}
\vspace{-0.05in}

%\begin{proof}
%    See the Appendix.
%\end{proof}

%So far we achieve an upper bound of $f(B)$ based on $\alpha(.)$. Next, we propose an upper bound of $\alpha(.)$.
%\begin{observation} \label{obs:hsize}
    %If we assume that there exists at least one edge which leads to an increase in balance on deleting from $H$, then, 
%    $\alpha(B_k) \le \frac{k - 1}{2}$.
%\end{observation}

%This is easy to prove since we need at least $2$ edges for one node to be counted in $\alpha(\{e_1, e_2, \cdots, e_k\})$. Now we are ready to prove our main result that the objective function is pseudo-submodular.
%The proof is in \cite{additional_proofs}. 
From Obs.~\ref{obs:singleEdgeDeletion}, it follows that only the deletion of a \emph{peripheral} edge may result in increase of balance. A peripheral edge has one endpoint  within $S(H)$ and the other outside $S(H)$. Owing to this result, hereon, we implicitly assume any edge being considered for deletion is a peripheral edge. Note, however, that following an edge deletion, the set of peripheral edges changes. Empowered with these observations, we next establish pseudo-submodularity.%In the next sections we provide the property of pseudo-submodularity and necessary approximation guarantees. 
%\vspace{-0.05in}
\subsubsection{Local Pseudo-submodularity}
\vspace{-0.05in}
\begin{defn}[Pseudo-submodularity \cite{santiago2020weakly}] Given a scalar $0< \gamma \leq 1$, a function $f$ is pseduo-submodular if  $\sum_{e \in R} [f(Q\cup \{e\}) -f(Q)] \geq \gamma [f(Q\cup R) - f(Q)]$ for any pair of disjoint sets $Q,R \subset \mathbb{C}$.
\end{defn}
\vspace{-0.05in}

Note that the pseudo-submodularity ratio $\gamma$ is a pessimistic bound over \emph{all} pairs of disjoint sets. Instead of using $\gamma$, we compute approximation bounds on a \emph{local} submodularity ratio \cite{santiago2020weakly} defined on two sets $Q,R$, i.e., a non-negative $\gamma_{Q, R}$ satisfying $\sum_{e \in R} [f(Q\cup \{e\}) -f(Q)] \geq \gamma_{Q, R} [f(Q\cup R) - f(Q)]$. It has been shown that using local bounds leads to significantly better guarantees \cite{santiago2020weakly}. First, we prove a lower bound for $\gamma_{Q, R}$ as follows:
\vspace{-0.05in}
\begin{thm} \label{thm:local_pseudo_submodular}
    For two disjoint sets $Q,R$,
    \vspace{-0.05in}
    \begin{equation*}
     \sum_{e \in R}{\left[f(Q\cup \{e\}) -f(Q)\right]}\geq \gamma_{Q, R}\left[f(Q\cup R) - f(Q)\right]
    \end{equation*}
     \vspace{-0.05in}

     where $ \gamma_{Q, R} \geq \frac{1}{1 + \frac{1}{4}\Delta(H_Q)(|R| - 1)}. $
\end{thm}
\vspace{-0.05in}

%\vspace{-0.05in}

\textsc{Proof.} See App.~\ref{sec:proof_pseduo_submodular}. 

This theorem proves a lower bound for $\gamma_{Q,R}$ for any disjoint sets $Q$ and $R$. %The approximation guarantees developed in our algorithms use this lower bound. %The lower bound is also tight (see the Appendix). 
Obs. \ref{obs:monotonic} and Thm.~\ref{thm:local_pseudo_submodular} show that the monotonicity and local pseudo-submodularity  holds for our objective function. We next leverage these properties to design a \emph{randomized greedy} algorithm with approximation guarantees.% respectively and now we use the following result proposed in
\vspace{-0.05in}
\subsection{Randomized Greedy (\rg)}
 %\cite{santiago2020weakly} to prove an approximation.

\begin{lem}[\cite{santiago2020weakly}]\label{lem:greedy_appx}
    %Let $f:2^{E} \rightarrow \mathbb{R}_+$ be a monotone set function. 
    Assuming $0 \le \gamma_i \le 1$ for $i \in \{0, 1, 2, \cdots, k-1\}$ so that $\sum_{e \in OPT} {[f(S_i \cup\{e\})-f(S_i))]} \geq \gamma_i \cdot [f(S_i \cup OPT)- f(S_i)]$ (local pseduo-submodularity) throughout the execution of the \rg algorithm, where $f$ is monotonic, $OPT$ denotes the optimal set of edges, and $S_i$ denotes the set of chosen elements after the $i$-th iteration (i.e. $|S_i| = i$); then \rg obtains an approximation of $1 - \exp{\left(-\frac{1}{k}\sum_{i=0}^{k-1}{\gamma_i}\right)}$ with a high probability. 
\end{lem}
\vspace{-0.05in}

We can directly apply this lemma in our setting. The \rg Algorithm is described as Algorithm \ref{alg:rgd}.% and achieve an approximation bound. 

\vspace{-0.05in}
\begin{thm} \label{thm:approx_RG}
For MBED, the \rg algorithm obtains an approximation of $1 - e^{-\gamma'}$, and $\gamma' \geq \frac{4}{4 + \Delta^{*} (b -1)}$
where $b$ and $\Delta^{*}$ denote the budget and the balance after deleting the optimal set of edges respectively.
\end{thm}
\vspace{-0.1in}

\begin{proof}
    Let us denote the optimal set of $b$ edges as $B^*$. By monotonicity, we get $\Delta(H_{S_0}) \le \Delta(H_{S_1}) \cdots \le \Delta(H_{S_{b-1}}) \le \Delta^{*}$. From Theorem \ref{thm:local_pseudo_submodular}, $\gamma_{S_i, B^*} \geq \frac{4}{4 + \Delta(H_{S_i})(|B^*| - 1)} \geq \frac{4}{4 + \Delta^{*} (b -1)}$.
    Now, by substituting $\gamma_i$ with $\gamma_{S_i, B^*}$ in Lem. \ref{lem:greedy_appx} we get the desired result.  
\end{proof}
\vspace{-0.05in}

\textbf{Improved Bounds:} The lower bound of $\gamma'$ in Thm.~\ref{thm:approx_RG} can be tighter. In particular, $\gamma' \geq \frac{4}{4 + \Delta^{RG} (b - 1)} $ where $\Delta^{RG}$ denotes the balance after deleting the solution set of $b$ edges produced by the \rg. The bound could be further improved as $\gamma' \geq \frac{4\psi^r}{4\psi^r + \Delta^{RG} (b - 1)} $ where $\psi^r$ is the summation of marginal gains of the elements in the optimal solution set over the solution set produced by \rg (see App. \ref{subsec:tigh_tbound_rg}). Table~\ref{tab:summary_approx} summarizes the additional lower bounds of $\gamma'$ (where the approximation guarantee is $1 - e^{-\gamma'}$) that can be derived on the \rg. 

\textbf{Implementation:} %Given an initial subgraph $H$ along with its current balanced subgraph $S(H)$, and a candidate set $\mathbb{C}$, our proposed algorithm (Algorithm~\ref{alg:rgd}) maximizes the balance in $H$ by deleting a set of $b$ edges from the candidate set. In particular, the 
Alg.~\ref{alg:rgd} first computes the set of peripheral edges of the initial balanced subgraph $S(H)$ (line 3). %This edge set is strategically chosen as it forms the set of contradictory edge-pairs $cep(H,x)$ for $x\in V$.
After that, for all peripheral candidate edges, $f(\{e\})$ is computed (lines $4-5$). %This can be done efficiently using Observation~\ref{obs:singleEdgeDeletion} and Lemma~\ref{lem:singleEdgeDeletion}. 
Using these values, the subset of peripheral edges of cardinality $b$ maximizing the sum of  $f(\{e\})$ is chosen and a random edge from this subset is selected for deletion (lines $6-8$). Following this edge deletion, the balanced subgraph $S(H)$ is updated to include the newly compatible nodes (line $9$). The peripheral edge set for the updated $S(H)$ is recomputed (line $3$) and this process continues in an iterative manner for $b$ iterations.

\begin{algorithm} [t]
\caption{Randomized Greedy}{\label{alg:rgd}}
{\scriptsize
\begin{algorithmic}[1] 
    \REQUIRE  The initial subgraph $H$, balanced subgraph $S(H)$, budget $b$, candidate set $\mathbb{C}$
 \ENSURE A set $B$ of $b$ edges
    % \STATE Compute the set of edges $\mathbb{C}^*$ on the periphery of $S(H)$ connecting it to the rest of $H$.
    \STATE $B \leftarrow \Phi$
    \FOR{$i=1$ to $b$}
        \STATE Compute the set of edges $\mathbb{C}^*$ on the periphery of $S(H)$ connecting to nodes in $H\setminus S(H)$.
        \FOR{$e \in \mathbb{C}^* \cap \mathbb{C}$}
            \STATE Compute $f(\{e\}) = \Delta(H_{\{e\}}) - \Delta(H)$.
            % = \Delta(H_{\{e\}}) - |V(S(H))|$
        \ENDFOR
        
        \STATE Find a subset $M^i \subseteq \mathbb{C}^*$ of size $b$ maximizing $\sum_{e \in M^i} {f(\{e\})}$.
        \STATE Select a uniformly random element $e_i$ from $M^i$.
        \STATE Delete $e_i$ from $H$, $B\leftarrow B\cup \{e_i\}$
        \STATE Update $S(H)$ to include the nodes from $\Delta(H_{\{e\}})$.
        %\STATE Recompute $\mathbb{C}^*$ as the set of edges on the periphery of the updated $S(H)$.
    \ENDFOR
    \RETURN {$B$}
\end{algorithmic}}
\end{algorithm}

% Theorem \ref{thm:weak_submodular} helps us designing a greedy and a randomized greedy algorithm with approximation guarantees. In particular, ...

%%%%%%%%%%%%%%%%%%%%%%%%%%% Algorithms %%%%%%%%%%%%
\input{algorithms}

%% file: algorithms.tex
\vspace{-0.05in}
\subsection{The Greedy Approach}
\label{sec:greedy}
%\label{sec:our_algos}
%In this section, we describe efficient versions of the \textit{randomized greedy algorithm} \cite{santiago2020weakly} and the  \textit{greedy algorithm} to solve MBED.
%proposed in \cite{santiago2020weakly} and propose an efficient version of it.% We also analyze its time complexity.

 The only difference with Alg.~\ref{alg:rgd} is that instead of choosing a random edge from the top $b$ edges with the highest sum of $f(\{e\})$s (lines $6-7$), the greedy algorithm (\greedy) chooses the edge with the highest $f(\{e\})$, i.e., $e_i = \argmax_{e\in \mathbb{C}^*}\{ f(\{e\})\}$.

\iffalse
One of the critical steps in Algorithm~\ref{alg:rgd} are the oracle calls for computation of function $f(\{e\})$ for all edges in the peripheral candidate edge set $(\mathbb{C}^*\cap \mathbb{C})$ for a given iteration. In the worst-case we make the following observation.
\begin{observation}
\label{obs:oracle_call_rgd}
Given budget $b$, candidate set $\mathbb{C}$, Algorithm~\ref{alg:rgd} requires $O(b|\mathbb{C}|)$ oracle calls in the worst-case.
\end{observation}
\fi

\textbf{Theoretical Bounds:}  
We derive the approximation of  \greedy in App.~\ref{sec:greed_approx}. Table \ref{tab:summary_approx} summarizes the different lower bounds of $\gamma'$ (where the approximation guarantee is $1 - e^{-\gamma'}$). %Since \textbf{RG} provides a superset of the bounds provided by \greedy, the tightness offered by \textbf{RG} is at least as good as \greedy.

\begin{table}[ht]
    \centering
    \vspace{-0.05in}
    {
    \begin{tabular}{c|c | c | c }
        \toprule
        Cases & I & II & III\\
        \midrule
        \rg & $\frac{4}{4 + \Delta^{*} (b - 1)}$ & $\frac{4}{4 + \Delta^{RG} (b - 1)}$ & $\frac{4\psi^r}{4\psi^r + \Delta^{RG} (b - 1)}$ \\ \hline
         \greedy & $\frac{4}{4 + \Delta^{*} (b - 1)}$ & $\frac{4}{4 + \Delta^{G} (b - 1)}$ & $\frac{4\psi^g}{4\psi^g + \Delta^{G} (b - 1)}$  \\ \hline
        
        %\bottomrule
    \end{tabular}}
    \caption{Lower bounds (higher is better) of $\gamma'$ produced by \rg and \greedy, where $\Delta^{*}$, $\Delta^{RG}$  and $\Delta^{G}$ denote the balance after deleting the optimal set of edges, the set produced by \rg and \greedy respectively. $\psi^r$ and $\psi^g$ are the summation of marginal gains of the elements in the optimal solution set over the solution set produced by \rg and \greedy respectively.}
    \label{tab:summary_approx}
    \vspace{-0.2in}
\end{table}
\vspace{-0.05in}
\subsection{Time Complexity} Alg.~\ref{alg:rgd} comprises of three main dominating parts with respect to the time complexity: (i) calls to compute function $f(\{e\})$ for all candidate edges, (ii) computing peripheral edge set (line 3) and (iii) finally updating the balanced subgraph $S(H)$ (line 9).  (i) Computing $f(\{e\})$: For each edge $e$ in the peripheral edge set, the computation of $f(\{e\})$ first checks if the corresponding vertex that is outside the balanced subgraph can be inducted inside on deletion of the given edge $e$. It checks the sign of all edges incident on the vertex, which on average consumes $O(d_{avg})$, where $d_{avg}$ is the average degree of a node in the graph. %From Observation~\ref{obs:singleEdgeDeletion}, a node can only be inducted if there are not more than one contradictory edge-pairs attached to it.
If the node is inducted, a breadth-first search (BFS) is performed to count its compatible neighbors that could be included in the newly balanced subgraph. So, each $f(\{e\})$ computation takes $O(|E|d_{avg})$ time. (ii-iii) For updating the balanced subgraph and corresponding peripheral edge set, a similar BFS is performed to find the vertices to be inducted in $S(H)$ and the incompatible edges during this search forms the peripheral edge set of the updated balanced subgraph. So, the overall time complexity of \rg is $O(b|\mathbb{C}||E|d_{avg})$ time. \greedy has the same  complexity.

%% file: 6_experiments.tex
\vspace{-0.05in}
\section{Experiments}
\label{sec:expts}
In this section, we benchmark the proposed algorithms and analyze their efficacy, efficiency and scalability.
\iffalse
\subsection{Ideas regarding experiments}
\begin{itemize}
    \item What is the balance of a community in real graph? Is it high or low to start with? Play with different number of communities and their balances. It would be an interesting study by itself.
    
    \item Does making a few changes in graph has a huge impact? If yes, it would be interesting. It is important to understand when it is easier to increase balance of a subgraph.
    
    \item Based on algorithms does any sign of edges get particular bias? Is it explainable?
    \item Can we say something about the nature of the graph based on the selected edges or the increase in balance?
    \item Quality experiments: varying budget on different algos.
    
    \item Running times
    \item varying the initial subgraph size with k cores 
    \item visualization
\end{itemize}
\fi

\begin{table}[t]
    \centering
   % \vspace{-0.15in}
    {\scriptsize
    \begin{tabular}{@{}lrrrrr@{}}
        \toprule
        Datasets & |V| & $|E_+ \cup E_-|$ & $\rho_{-}$ & $|V(G^*)|$ & $|\Delta(G^*)|$\\
        \midrule
        BitcoinAlpha & 4k & 14k & 0.09 & 3772 & 2903\\
        BitcoinOTC & 6k & 21k & 0.15 & 5872 & 4487 \\
        Chess & 7k & 32k & 0.42 & 6601 & 3477\\
        WikiElections & 7k & 100k & 0.22 & 7066 & 3857 \\
        Slashdot & 82k & 498k & 0.23 & 82052 & 51486\\
        WikiConflict & 118k & 1.4M & 0.62 & 96243 & 53542 \\
        Epinions & 131k & 708k & 0.17 & 119070 & 81385\\
        WikiPolitics & 138k & 712k & 0.12 & 137713 & 68037 \\
        
        \bottomrule
    \end{tabular}}
    \caption{Description of Datasets: $G^*$ and $\Delta(G^*)$  denote the largest connected component (LCC) and  the maximum balanced subgraph of LCC respectively in graph $G$. $\rho_{-}=\frac{|E_-|}{|E_+ \cup E_-|}$ denotes the proportion of negative edges in the graph.}
    \label{tab:dataset_statistics}
    \vspace{-0.15in}
\end{table}

%%%%%%%%%%%%%%%%%%%%%%%%%%% Quality varying budget %%%%
\begin{figure*}[ht]
    \centering
    \captionsetup[subfigure]{labelfont={normalsize,bf},textfont={normalsize,bf}}
    \vspace{-0.30in}
    \subfloat[NYC Cab][BitcoinAlpha]{\includegraphics[width=0.22\textwidth]{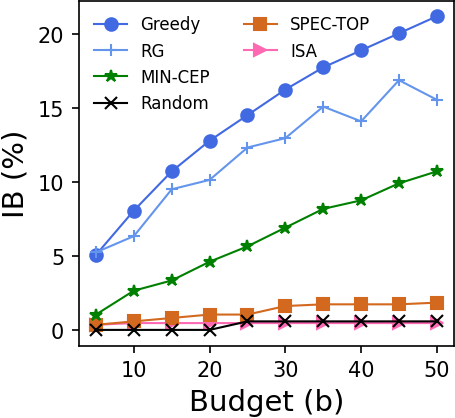}}\hfill
    \subfloat[NYC Cab][BitcoinOTC]{\includegraphics[width=0.23\textwidth]{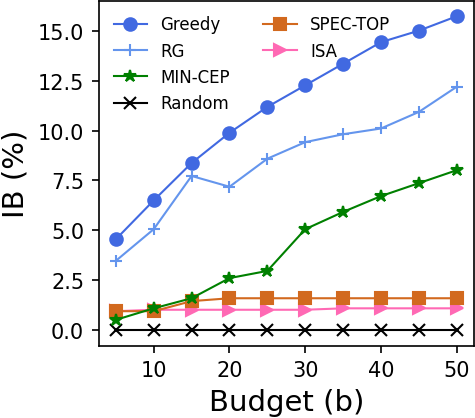}}\hfill
    \subfloat[NYC Cab][Chess]{\includegraphics[width=0.215\textwidth]{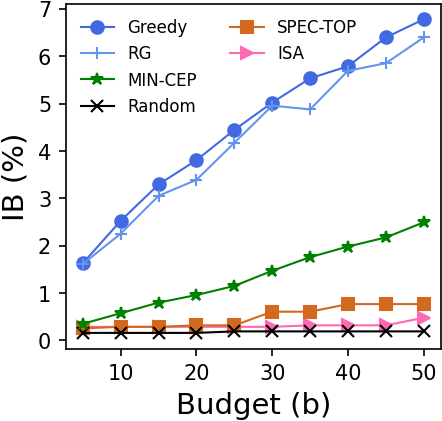}}\hfill
    \subfloat[NYC Cab][WikiElections]{\includegraphics[width=0.215\textwidth]{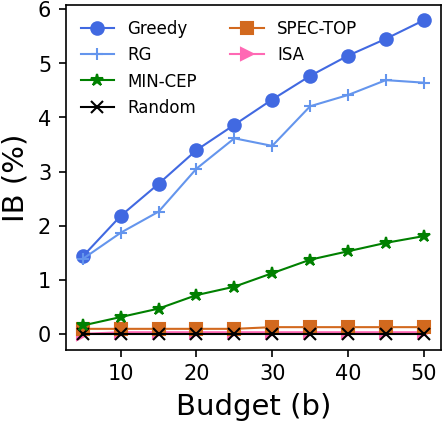}}\\
        \vspace{-0.15in}
    \subfloat[NYC Cab][Epinions]{\includegraphics[width=0.22\textwidth]{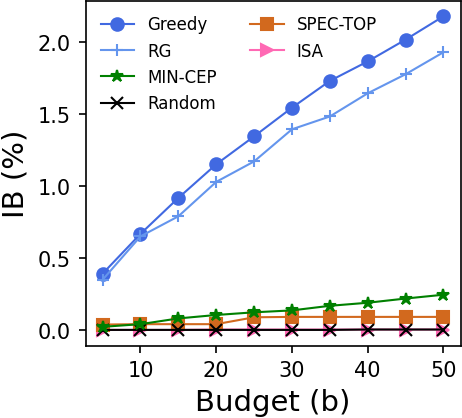}}\hfill
    \subfloat[NYC Cab][Slashdot]{\includegraphics[width=0.22\textwidth]{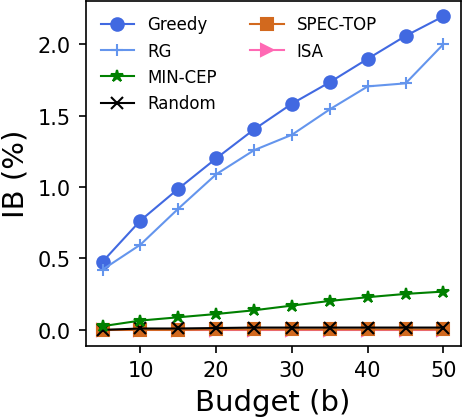}}\hfill
    \subfloat[NYC Cab][WikiConflict]{\includegraphics[width=0.22\textwidth]{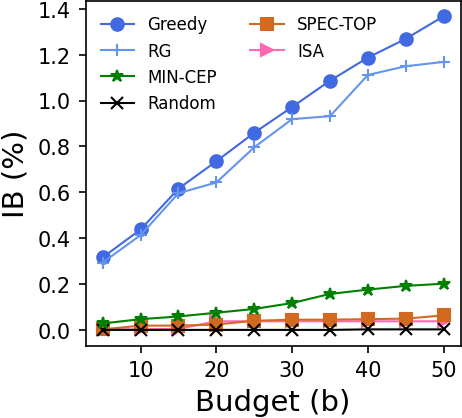}}\hfill
    \subfloat[WikiPolitics]{\includegraphics[width=0.22\textwidth]{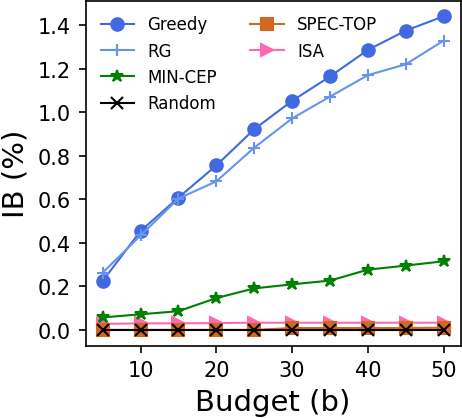}}
        \vspace{-0.15in}
    \caption{ Impact of budget on IB\% (Eq.~\ref{eq:performance}). \rg and \greedy are superior by up to $7$ times than the closest baseline (\textsc{Min-Cep}).}
\label{fig:quality_varying_budget}
    \vspace{-0.01in}
\end{figure*}

%%%%%%%%%%%%%%%%%%%%%%%%%%%%%%%%%%%%%%%%%%%%%%%%%%
\vspace{-0.10in}
\subsection{Experimental Setup}
All algorithms have been implemented in Python $3.6.9$ on a Ubuntu $18.04$ PC with a $2.1$ GHz
Intel\textsuperscript{\textregistered} Xeon\textsuperscript{\textregistered}
Platinum $8160$ processor, $256$ GB RAM and a $7200$ RPM, $8.5$ TB disk. The codebase is available online\footnote{https://github.com/Ksartik/MBED}.
\vspace{-0.05in}
\subsubsection{Datasets}
We use publicly available signed networks from \url{http://konect.cc}. Table~\ref{tab:dataset_statistics} summarizes the dataset statistics. Each of these models polarized (signed) social interactions. BitcoinOTC, BitcoinAlpha, Epinions are trust/distrust networks on the two respective Bitcoin trading platforms and an online product rating site respectively. Chess represents the chess games' results with edges being positive if white won and negative otherwise. Slashdot comprises the friend/foe relations on the news site Slashdot. The edges in WikiConflict represent the positive/negative conflicts on the Wikipedia. WikiPolitics contains interpreted interactions between editors of political articles on Wikipedia. WikiElections connects Wikipedia users who voted for/against each other. We ignore the direction of the edges in the directed graphs and remove any loops and multi-edges.%Epinions is a trust/distrust network of an online product rating site by that name. Slashdot comprises the friend/foe relations on the news site Slashdot. The edges in WikiConflict represent the positive/negative conflicts on the Wikipedia. WikiPolitics contains interpreted interactions between editors of political articles on Wikipedia. WikiElections connects Wikipedia users who voted for/against each other. 
% Cloister connects monks with ratings given to them by one another. Congress contains favorable/unfavorable mentions made by politicians of each other. 

%Note that many of these datasets are directed while our problem is defined for an undirected graph. Therefore, we ignore the direction of the edges in the directed graphs and remove any loops and multi-edges. 

%%%%%%%%%%%%%%%%%%%%%%%%%%% Quality varying large budget %%%%

\vspace{-0.05in}
\subsubsection{Baselines}
Besides \textsc{Greedy} and Randomized Greedy (\textsc{\rg}), we consider the following baselines:
\begin{itemize}
\item \textsc{Spec-Top}: In \S\ref{sec:spectral_top}, we design a spectral approach using an upperbound of the minimum eigen value of the Laplacian. 
    \item \textsc{Isa}: Alg.~\ref{alg:spectral_edge} describes this baseline, which is based on perturbation theory. We only consider the peripheral edges as the candidates.% edges in the periphery of $S(H)$ where $H$ is initial .
    \item \textsc{Random}: We randomly delete $b$ edges from the periphery of $S(H)$, where $H$ is the initial given subgraph. 
    \item \textsc{Min-Cep}: Obs. \ref{obs:singleEdgeDeletion} shows that an edge ($e$) deletion associated with a node $x$ is favorable if $|cep(H_{\{e\}}, x)| = 0$. Thus, we iteratively delete the peripheral edge minimizing $|cep(H_{\{e\}}, x)|$.
\end{itemize}

\vspace{-0.05in}
\subsubsection{Parameters:}
The default input subgraph $H$ is the largest connected component (LCC) of the signed graph. We find the initial maximum balanced graph $S(H)$ using TIMBAL~\cite{ordozgoiti2020finding}. Table~\ref{tab:dataset_statistics} lists the size of the LCC and its balance in each of the datasets. In addition, for some experiments, we also use $k$-core structures that are well-known for community discovery~\cite{peng2014accelerating}. The set of candidate edges $\mathbb{C}$ is set to all edges in $H$. The budget $b$ is varied in each experiment.
\vspace{-0.05in}
\subsubsection{Performance Metric: } The quality of a solution (edge) set $B$ for a given subgraph $H$ is defined as the percentage of nodes that gets included in the balanced subgraph after the deletion of $B$. 

\vspace{-0.10in}
\begin{equation}\label{eq:performance}
    IB(B,H) (\%)= \frac{\Delta(H_B)-\Delta(H)}{|H|-\Delta(H)} \times 100.%{|H|-\Delta(H)}
\end{equation}
\vspace{-0.10in}

\begin{figure*}[ht]
\vspace{-0.30in}
    \centering
    \captionsetup[subfigure]{labelfont={normalsize,bf},textfont={normalsize,bf}}
    \subfloat[NYC Cab][Epinions]{\includegraphics[width=0.22\textwidth]{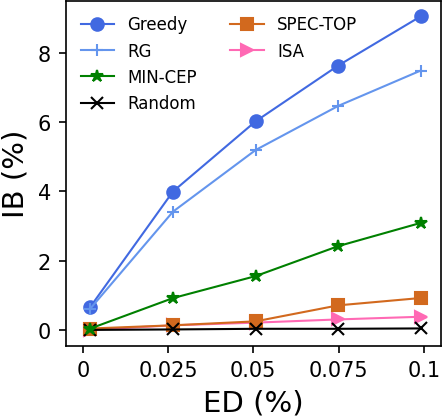}}\hfill
    \subfloat[NYC Cab][Slashdot]{\includegraphics[width=0.22\textwidth]{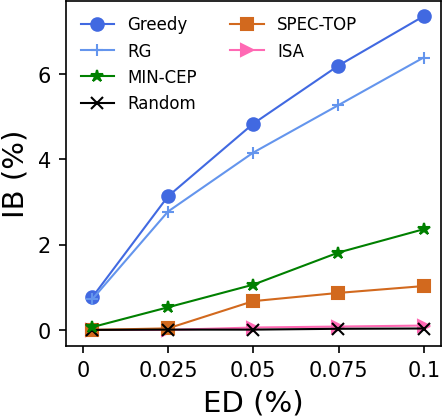}}\hfill
    \subfloat[NYC Cab][WikiConflict]{\includegraphics[width=0.22\textwidth]{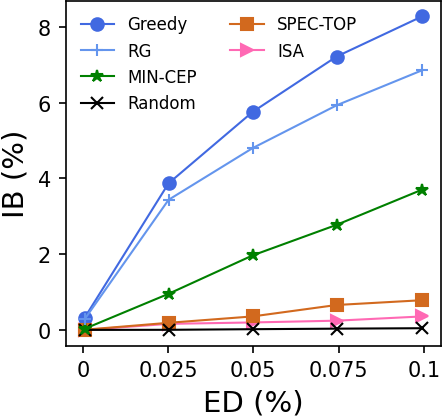}}\hfill
    \subfloat[NYC Cab][WikiPolitics]{\includegraphics[width=0.22\textwidth]{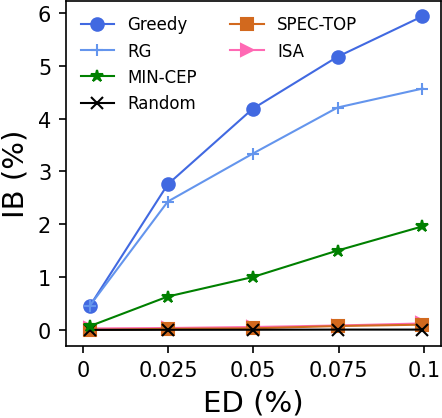}}\hfill
    \vspace{-0.15in}
    \caption{The quality of all methods with large budgets (ED implies the fraction of edge deletions) in four large datasets.} %Our methods outperform the best baseline by up to $6\%$ while achieving $9\%$ increase in balance with deleting only $0.1\%$ of edges.}
\label{fig:quality_varying_large_budget}
\end{figure*}
\begin{figure*}[ht]
\vspace{-0.25in}
    \centering
    \captionsetup[subfigure]{labelfont={normalsize,bf},textfont={normalsize,bf}}
    % \subfloat[NYC Cab][BitcoinAlpha]{\includegraphics[width=0.22\textwidth]{kcore results/BitcoinAlpha.png}}
    % \subfloat[NYC Cab][BitcoinOTC]{\includegraphics[width=0.22\textwidth]{kcore results/BitcoinOTC.png}}
    % \subfloat[NYC Cab][Chess]{\includegraphics[width=0.22\textwidth]{kcore results/Chess.png}}
    % \subfloat[NYC Cab][WikiElections]{\includegraphics[width=0.22\textwidth]{kcore results/WikipediaElections.png}}\\
    \subfloat[NYC Cab][Epinions]{\includegraphics[width=0.21\textwidth]{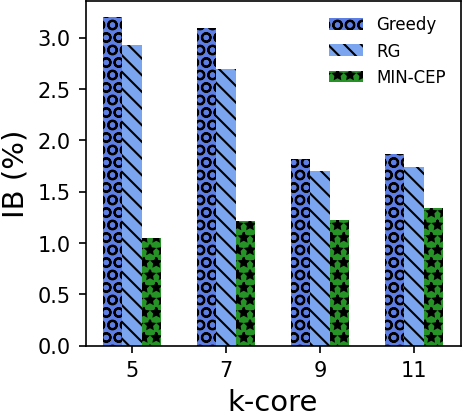}}\hfill
    \subfloat[NYC Cab][Slashdot]{\includegraphics[width=0.22\textwidth]{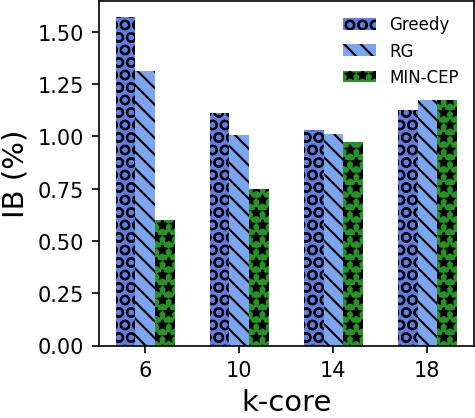}}\hfill
    \subfloat[NYC Cab][WikiConflict]{\includegraphics[width=0.22\textwidth]{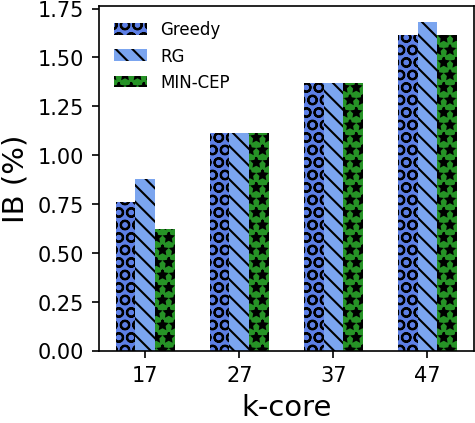}}\hfill
    \subfloat[NYC Cab][WikiPolitics]{\includegraphics[width=0.21\textwidth]{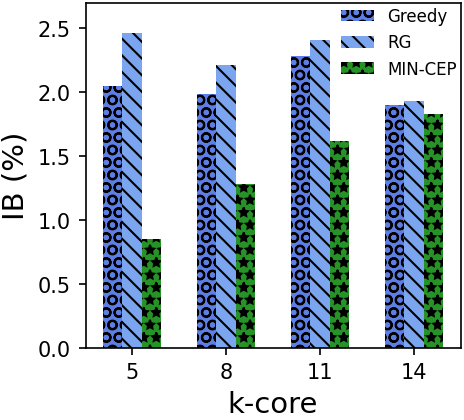}}\hfill
    \vspace{-0.15in}
    \caption{Increase in the balance when the input subgraph is a $k$-core. Results are shown against varying values of $k$ for $b=50$.}
\label{fig:quality_varying_kcore}
\vspace{-0.05in}
\end{figure*}

\vspace{-0.10in}
\subsection{Efficacy and Efficiency}
%In this section, we test our algorithm for its efficacy and efficiency. We demonstrate the balance and r achieved by our algorithm for different budgets. %Moreover, we also check if our algorithm is scalable to larger graphs.
%In particular, the performance metrics to compare all the algorithms are $IB(S,H)$ (Eq. \ref{eq:performance}) and running time in seconds where $S$ is the solution set produced by the algorithm and $H$ is the given initial subgraph (the largest connected component in these experiments). 

\subsubsection{Small budget on all datasets: }Fig.~\ref{fig:quality_varying_budget} shows the percentage increase in balance (IB) for eight datasets achieved by each algorithm. \textsc{Greedy} and \textsc{\rg} outperform all the baselines by up to 12\%. Besides having approximation guarantees (Thms. \ref{thm:approx_RG} and \ref{thm:approx_greedy}), \greedy and \rg directly optimize the objective function in an iterative fashion. In contrast, the baselines choose solution edges depending on other criterion. In particular, the spectral methods \textsc{Isa} and \textsc{Spec-Top} do not perform well since it chooses edges based on an upper-bound to minimize the minimum eigenvalue of the corresponding Laplacian. Though the balanced graph has minimum eigenvalue of the Laplacian as $0$, the rate at which the edge deletions move towards achieving it, might still be low. We also observe that \greedy, in general, performs better than \rg. It would be wrong, however, to draw the conclusion that \greedy is always better. In subsequent experiments where we choose $k$-cores as the input subgraphs, we will see that \rg performs better. We will revisit the topic of \greedy vs \rg while discussing that experiment.

\vspace{-0.05in}
\subsubsection{Larger budget on large datasets: }To further demonstrate the efficacy of our methods we vary the budget as a function of $\mathbb{C}$. i.e., all edges in $H$. Fig.~\ref{fig:quality_varying_large_budget} shows the percentage increase in balance (IB) for the four largest datasets. Consistent with previous experiments, \rg and \greedy outperform all baselines (better by up to $6\%$ points). More interestingly, we observe that a substantial increase in balance is feasible ($9\%$ or up to $4000$ nodes) by deleting only $0.1\%$ of edges ($\approx500$ edges). In other words, improvement in balance-dependent community functions, such as team performance or stability, may be significantly improved through minor adjustments to the network.

\vspace{-0.05in}
\subsubsection{Scalability: }Table \ref{tab:rt_varying_budget} shows the running times of all  algorithms against budget in the three largest datasets. Although \rg and \greedy are slower than the other baselines, they finish within a few minutes even on a million edges' network. Thus, scalability to large networks is not a concern. A more interesting behavior is witnessed in the  correlation between efficacy and efficiency. More specifically, we observe that the better performance of an algorithm in IB\%, the higher is its running time. When an algorithm performs better, it means in each iteration, the algorithm  produces a larger cascading impact following an edge deletion. Higher cascading impact leads to a larger number of new peripheral edges coming into consideration. Consequently, the running time goes up. 

\begin{table}[b]
%\vspactye{-0.10in}
{\scriptsize
\begin{tabular}{l|rrr|rrr|rrr}
\toprule
    % \multirow{2}{*}{\diagbox{\small{Algorithm}}{\small{Dataset}}}&
    % \diagbox{Algo}{Data\\b} &
    & \multicolumn{3}{c|}{Epinions} & \multicolumn{3}{c|}{WikiPolitics} & \multicolumn{3}{c}{WikiConflict} \\
    \backslashbox{Method}{Budget}& 10 & 30  & 50  & 10   & 30   & 50   & 10  & 30   & 50   \\
\midrule
\textsc{Isa} & 2  & 6 & 11  & 2  & 7  & 12 & 4 & 12 & 19 \\
\textsc{Spec-Top} & 2 & 6 & 9 & 3 & 9 & 15 & 3 & 7 & 11 \\
\textsc{Min-Cep} & 4 & 5 & 6 & 4 & 6 & 7 & 10 & 12 & 14 \\
\textsc{\rg} & 6 & 7 & 9 & 7 & 9 & 10 & 13  & 16  & 18  \\
\textsc{Greedy} & 7 & 13 & 18 & 9 & 18  & 25  & 15  & 22  & 28  \\
\bottomrule
\end{tabular}}
\caption{Running times in minutes of the algorithms varying budget on largest available datasets.}
\label{tab:rt_varying_budget}
\vspace{-0.20in}
\end{table}

\iffalse
\begin{table}[b]
\vspace{-0.20in}
\begin{tabular}{l|rrr|rrr|rrr}
\toprule
    % \multirow{2}{*}{\diagbox{\small{Algorithm}}{\small{Dataset}}}&
    % \diagbox{Algo}{Data\\b} &
    & \multicolumn{3}{c|}{Epinions} & \multicolumn{3}{c|}{WikiPolitics} & \multicolumn{3}{c}{WikiConflict} \\
    \backslashbox{Method}{Budget}& 10 & 30  & 50  & 10   & 30   & 50   & 10  & 30   & 50   \\
\midrule
\textsc{Isa} & 2  & 6 & 11  & 2  & 7  & 12 & 4 & 12 & 19 \\
\textsc{Spec-Top} & 2 & 6 & 9 & 3 & 9 & 15 & 3 & 7 & 11 \\
\textsc{Min-Cep} & 36 & 110 & 187 & 41 & 123 & 205 & 91 & 274 & 455 \\
\textsc{Greedy} & 41 & 120 & 200 & 47 & 138  & 229  & 99  & 294  & 488  \\
\textsc{RG} & 40 & 127 & 198 & 46 & 135  & 227  & 97  & 290  & 490  \\
\bottomrule
\end{tabular}
\caption{Running times in minutes of the algorithms varying budget on largest available datasets.}
\label{tab:rt_varying_budget}
\vspace{-0.20in}
\end{table}
\fi

\iffalse
\begin{table}[!h]
    \centering
    \begin{tabular}{@{}lrr@{}}
        \toprule
        Datasets & |V(H)| & $|\Delta(H)|$ \\
        \midrule
        BitcoinAlpha & 3772 & 2903 \\
        BitcoinOTC & 5872 & 4487  \\
        Chess & 6601 & 3477 \\
        WikiElections & 7066 & 3857 \\
        Slashdot & 82052 & 51486 \\
        WikiConflict & 96243 & 53542 \\
        Epinions & 119070 & 81385 \\
        WikiPolitics & 137713 & 68037 \\
        \bottomrule
    \end{tabular}
    \caption{Caption}
    \label{tab:my_label}
\end{table}
\fi

\vspace{-0.10in}
\subsection{Impact of Community Density}
In this experiment, we systematically vary the density of the input community $H$ and analyze its impact on the performance. To control the density of $H$, we use $k$-core \cite{zhang2017finding} as the input subgraph. As $k$ increases, $H$ gets denser. Table~\ref{tab:kcore_statistics} shows the maximum and minimum $k$-core sizes along with their balance for each dataset. We vary the value of $k$ depending on the $k$-core distribution of the graph. As high $k$-cores contain fewer nodes, the highest value of $k$ is chosen such that the size of the $k$-core is at least $10\%$ of the original graph size in terms of number of nodes.

Fig.~\ref{fig:quality_varying_kcore} presents the results. In this section, we only consider the three best-performing algorithms of \greedy, \rg and \textsc{Min-Cep}. \greedy and \rg continue to be the best performers. Another interesting behavior we observe is that, the higher the $k$, and therefore density, the smaller is the gap between \greedy and \rg. In some cases, \rg performs better than \greedy. This behavior is a direct consequence of how \rg and \greedy operates. \greedy deterministically chooses the edge with the highest marginal gain. Consequently, when the gradient of the marginal gains in the sorted order is high, choosing the highest edge produces a good result. However, when the gradient is small and several edges provide similarly high marginal gains, \rg performs better.

%\subsection{Non-Monotonic Case: Power of Randomized Greedy}
\vspace{-0.05in}
\subsection{Visualizations on Bitcoin Network}
In the next experiment, we visually inspect the impact of edge deletions on increasing balance in the BitcoinOTC data. Fig.~\ref{fig:visBitcoin} presents the gradual increase in the size of the balanced component following $5$ and $10$ edge deletions. It shows that: (1) both positive and negative edges are chosen for deletion, and (2) there may be significant cascading impact of a single deletion (as visible in the appearance of several new green squares in Fig.~\ref{fig:visBitcoin}(c)).%shows how deleting a few edges are adding more nodes to the initial balanced portion of the graph. We apply \textsc{Greedy} algorithm with budget of $5$ and $10$ edges. This shows both positive (red and dashed) and negative edges (blue and dashed) appear in the selection. It also depicts that an edge deletion not only brings one node but a balanced portion that is attached to it. 

%% file: 7_conclusion.tex
%\vspace{-0.05in}
\section{Conclusions}
In this paper, we studied the problem of maximizing the balance in signed networks via edge deletion. While existing studies have focused primarily on finding the largest balanced subgraph, we adopted a network design approach to improve balance inside a subgraph. We proved that the problem is NP-hard, non-submodular, and non-monotonic. To overcome the resultant computational challenges, we designed an efficient heuristic based on the relation of Laplacian eigenvalues with the balance in corresponding signed graphs. Since these heuristics do not exhibit approximation guarantees, we leverage pseudo-submodularity of the objective function to design greedy algorithms with provable approximation guarantees. %further improved the a greedy, along with its randomized version, for the problem with approximation guarantees. 
 Through an extensive set of experiments, we showed that the proposed approximation algorithms outperform the baseline algorithms while being scalable to large graphs. An interesting future direction would be to explore alternative network design mechanisms such as node deletion and edge-sign flips to improve balance. From a theoretical perspective, we also aim to investigate the parameterized complexity of balance-related design problems.

\begin{table}[b]
    \centering
    \vspace{-0.05in}
    {\scriptsize
    \begin{tabular}{@{}lrrrr@{}}
        \toprule
        Datasets & $|V(H_{k_{min}})|$ & $|\Delta(H_{k_{min}})|$ & $|V(H_{k_{max}})|$ & $|\Delta(H_{k_{max}})|$\\
        \midrule
        Epinions & 26k & 20k & 13k & 10k\\
        Slashdot & 23k & 14k & 8k & 4k\\
        WikiConflict & 25k & 17k & 12k & 9k \\
        WikiPolitics & 37k & 30k & 14k & 11k \\
        \bottomrule
    \end{tabular}}
    \caption{Sizes of the $k$-core ($H_k$) corresponding to the minimum ($H_{k_{min}}$) and maximum ($H_{k_{max}}$) values of $k$ considered for each dataset in  Fig.~\ref{fig:quality_varying_kcore}.} %Note that the max. value is chosen such that the size of the $k$-core is greater than at least $10\%$ of the original graph size and minimum is chosen by making 5 divisions based on that.}
    \label{tab:kcore_statistics}
    \vspace{-0.30in}
\end{table}

\begin{figure}[b]
%\vspace{-0.25in}
    \centering
    \captionsetup[subfigure]{labelfont={normalsize,bf},textfont={normalsize,bf}}
    \subfloat[NYC Cab][Initial subgraph] {\includegraphics[width=0.15\textwidth]{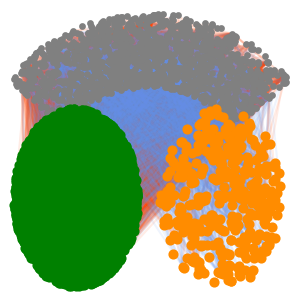}\label{fig:visBitcoin_a}}
    \subfloat[NYC Cab][{\small After $5$ deletions}]{\includegraphics[width=0.15\textwidth]{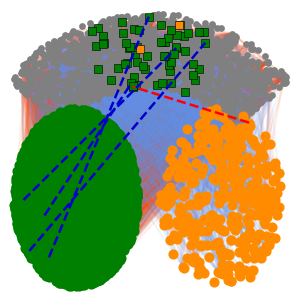}\label{fig:visBitcoin_b}}
    \subfloat[NYC Cab][{\small After $10$ deletions}]{\includegraphics[width=0.15\textwidth]{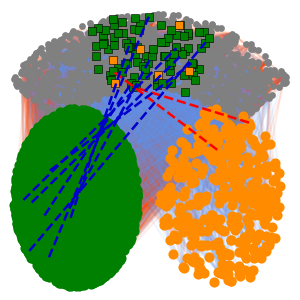}\label{fig:visBitcoin_c}}
    \vspace{-0.15in}
    \caption{Visualization of the impact of edge deletions by \greedy. Green and orange denote the two partitions of the balanced subgraph $S(H)$; grey denotes the component outside $S(H)$. The solid red and  blue edges are positive and negative edges, respectively, while the dashed edges in (b) and (c) are the ones being deleted. (b) and (c) also show the new components being added to the balanced portion through green and orange squares.}% (a) It shows the initial subgraph with the grey nodes outside of the initial balanced subgraph. After deleting (b) 5 edges and (c) 10 edges some square nodes are added to one of the partitions with corresponding colors.}
\label{fig:visBitcoin}
\vspace{-0.10in}
\end{figure}

%% file: 7_1_paper_appendix.tex
\vspace{-0.10in}
\section{Appendix}
\label{sec:paper_appendix}
%\appendix
%\renewcommand{\thesubsection}{\Alph{subsection}}
% \vspace{-0.05in}

\iffalse
\subsection{Proof of Lemma \ref{lemma:modular}}
\label{app:modular}

We denote $g_X(Y)$ as the marginal gain of the set of edges $Y$ over the set $X$, i.e., $g_X(Y) = g(X\cup Y)-g(X)$. To prove modularity, we need to show $g_X(Y)=\sum_{e \in Y} {g_X(e)}$, i.e. the marginal gain of the set of $Y$ over $X$ is the summation of the marginal gains of each individual in $Y$ over $X$ for any $X,Y$.

\vspace{-0.05in}
\begin{proof} We can write $g_X(Y)$ as follows.

\vspace{-0.05in}
    \begin{align*}
        g_X(Y) &= - \sum_{(i, j) \in X \cup Y} {{(\bm{v}_i - \sigma(i, j)\bm{v}_j)}^2} + \sum_{(i, j) \in X} {{(\bm{v}_i - \sigma(i, j)\bm{v}_j)}^2} \\
        &= -\sum_{(i, j) \in Y} {{(\bm{v}_i - \sigma(i, j)\bm{v}_j)}^2} 
        = \sum_{e \in Y} {g_X(e)}
    \end{align*}
\end{proof}
\fi

% \vspace{-0.30in}
\subsection{Proof of Lemma \ref{lemma:edge_perturb}} 
\label{sec:proof_edge_perturb}
\vspace{-0.05in}
\begin{proof}
Given a signed graph $G=(V,E,\sigma)$, a subgraph $H$, let $\lambda_i$, $\tilde{\lambda_i}$ be the eigenvalues of $L(H)$ and the perturbed matrix $\hat{L}(H)$ (after single edge $(i,j)$ deletion) respectively where $\lambda_1\leq\lambda_2\leq\cdots\leq\lambda_m$.% and $\tilde{\lambda_i}$ be the eigenvalues of the perturbed matrix $\hat{L}(H)$ after single edge $(i,j)\in E$ deletion.

We have $\hat{L}(H)=L(H)+P$, and perturbation matrix $P=\bar{D}+S$, where $\bar{D}$ is a diagonal matrix with $\bar{D}_{ii}=\bar{D}_{jj}=-1$ and $0$ otherwise.
$S_{ij}=S_{ji}=\sigma(i,j)$ for the perturbed edge $(i,j)\in E$ and otherwise $0$. Given $\bm{v}$ as the unit eigenvector corresponding to $\lambda_1$ we have,
\vspace{-0.05in}
\begin{align*}
\bm{v}^T\bar{D}\bm{v}=\sum_{k=i,j}\bar{D}_{kk}v^2_k, \text{ and }
\bm{v}^TS\bm{v}=\sum_{i,j}\sigma(i,j)v_iv_j
\label{eq:s_mat}
\end{align*}

\vspace{-0.05in}
From the first-order matrix perturbation theory (see p. 183 \cite{Stewart-Book:1990}),

\vspace{-0.05in}
\begin{align*}
%\label{eq:final_value} 
\tilde{\lambda}_1&=\lambda_1+\bm{v}^TP\bm{v}+O(||P||^2_F)\nonumber
=\lambda_1+\bm{v}^T\bar{D}\bm{v}+\bm{v}^TS\bm{v}+O(||P||^2_F)\nonumber\\
% &=\lambda_1-\sum_{k=i,j}\bm{v}^2_k+\sum_{i,j}\sigma(i,j)\bm{v}_i\bm{v}_j+O(2(1+\sigma(i,j)^2))\\
&=\lambda_1-\sum_{k=i,j}v^2_k+\sum_{i,j}\sigma(i,j)v_iv_j+O(1)\\
&=\lambda_1-v_i(v_i-\sigma(i,j)v_j)-v_j(v_j-\sigma(j,i)v_i)+O(1)\\
&=\lambda_1-\left(v_i-\sigma(i,j)v_j\right)^2+O(1)
\end{align*}

\vspace{-0.05in}

%\end{equation*}
Now, to show that $\tilde{\lambda_1}(H)$ is indeed the smallest eigenvalue of $\hat{L}(H)$, using matrix perturbation theory (p. 203 \cite{Stewart-Book:1990}), we have

\vspace{-0.05in}

\begin{align*}
    \tilde{\lambda_1}&\leq \lambda_1+||P||_2 \leq \lambda_1+||P||_F\leq \lambda_1+2\\
    \tilde{\lambda_i}&\geq \lambda_i-||P||_2 \geq \lambda_i-||P||_F\geq \lambda_i-2,  ~(i\geq 2)
 \end{align*}
 
 \vspace{-0.05in}

Since the spectral gap $\delta=\lambda_2-\lambda_1\geq 4$, we have $\tilde{\lambda_i}\geq\tilde{\lambda_1}$. So, we have $\tilde{\lambda_1}=\hat{\lambda}_1$ is the smallest eigenvalue of $\hat{L}(H)$.
\end{proof}

\vspace{-0.20in}
\subsection{Details for proof for Theorem \ref{thm:local_pseudo_submodular}}
\label{sec:proof_pseduo_submodular}

Before proving Thm. \ref{thm:local_pseudo_submodular}, we derive a few results. Let $\hat{\mathbb{V}}$ be the node set that gets added in the maximum balanced subgraph $S(H)$ after deleting $B$ edges. We know that $\forall u \in \hat{\mathbb{V}}$ there exists $v\in S(H)$ such that $(u,v)\in B$. The inclusion of one node may lead to including more nodes in the balanced portion. Let $C_u$ be the size of component that gets added with $u\in \hat{\mathbb{V}}$ and $C^* = \max\{C_u,\: u\in \hat{\mathbb{V}}\}$. %W, the upper bound is as follows:
\begin{observation} \label{obs:c_star_UB}
\vspace{-0.05in}
\begin{equation}\label{eq:C_star}
    C^* + 1 \le \frac{\Delta(H)}{2}
    \vspace{-0.05in}
\end{equation}

\end{observation}
\vspace{-0.10in}
\textsc{Proof by Contradiction.}
%    See \cite{additional_proofs} (Section \ref{sec:c_star_UB}).\end{proof}
   % \begin{proof}
   %From Observation~\ref{obs:singleEdgeDeletion}, when a node $u$ comes into the modified balanced subgraph, it can further bring $C_u$. It will have maximum size $\frac{\Delta(H) - 2}{2}$. This is because 
   If $|C_u|>\frac{\Delta(H) - 2}{2}$, then the initial $S(H)$ would consist of the larger among $V_1$ and $V_2$ (which would be at least of size $\frac{\Delta(H)}{2}$) along with $\{u\} \cup C_u$. %\end{proof}
%\vspace{-0.05in}

\textbf{Choice of $\alpha(B)$ and Peripheral Edges (PE):} Let $\alpha(B)$ be the number of nodes $x$ satisfying: (1) $|cep(H_{\{e\}}, x)| > 0,\forall e\in B$ and (2) $|cep(H_{Y}, x)| = 0$ for some subset $Y \subseteq B, Y\neq \emptyset$. We use Obs. \ref{obs:singleEdgeDeletion} to restrict the edge set $B$ to always belong to the \emph{periphery} of the current balanced subgraph.  An upperbound of $f(B)$ is as follows. 
\vspace{-0.05in}
\begin{lem}
\label{lemma:induction}
    \begin{equation}
    %\vspace{-0.30in}
        f(B) \le \sum_{i=1}^{b}{f(\{e_i\})} + (C^* + 1) \alpha(B).
    \end{equation}

\end{lem}
%[[$\alpha$ and $C^*$ need to be defined inside the statement of this lemma. AB.]]
\vspace{-0.05in}

\begin{proof}
   This is proved using induction (Sec. \ref{sec:proof_lemmma_induction}).
\end{proof}
\vspace{-0.10in}

\subsubsection{Final proof for Theorem \ref{thm:local_pseudo_submodular}}
\vspace{-0.05in}
\begin{proof}
    Note that $f(Q \cup R) - f(Q) = \Delta(H_{Q \cup R}) - \Delta(H_{Q})$. We can write this as $\Delta(H'_R) - \Delta(H')$, where $H' = H_Q$. That means marginal gain in balance of deleting the set $R$ over $Q$ is same as the marginal gain in balance of deleting the set $R$ from $H_Q$. We can thus use $f'(R) = \Delta(H'_R) - \Delta(H')$ in place of $f$. Thus, by Lem.~\ref{lemma:induction}:
    
    \vspace{-0.05in}
    \begin{equation}\label{eq:f_dash}
        f'(B) \le \sum_{i=1}^{b}{f'(\{e_i\})} + (C^* + 1)  \alpha(B) %\left(\text{where $C^*=C(H')$ or $C(H_Q)$}\right)
    \end{equation}
    \vspace{-0.05in}

    where  $C^*$ and $\alpha$ are defined accordingly to new initial subgraph $H' = H_Q$.
    Next, we propose an upper bound of $\alpha(.)$ as follows:
    \begin{equation}\label{eq:alpha_ub}
    \alpha(B) \le \frac{|B| - 1}{2}
    \end{equation}
    This is true since we need at least two edges for one node to be counted in $\alpha(B)$. %Now we are ready to prove our main result that the objective function is pseudo-submodular. 
    \vspace{-0.05in}
\begin{equation*}
\begin{split}
\text{Now we have,  }\frac{\sum_{e \in R}{[f(Q \cup \{e\})-f(Q)]}}{f(Q \cup R)- f(Q)} = \frac{\sum_{e \in R}{f'(\{e\})}}{f'(R)} \\
\geq \frac{1}{1 + \frac{\alpha(R)(C^* + 1)}{\sum_{e \in R}{f'(\{e\})}}} \left(\text{ Replace $f'(R)$ using Eq. \ref{eq:f_dash}}\right) \\
\end{split}
\end{equation*}
\begin{equation*}
    \begin{split}
\geq \frac{1}{1 + \frac{\frac{|R|-1}{2}(C^* + 1)}{\sum_{e \in R}{f'(\{e\})}}} \left(\text{ Using the upper bound of $\alpha$ in Eq. \ref{eq:alpha_ub} }\right) \\
\geq \frac{1}{1 + \frac{1}{4}\Delta(H_Q)(|R| - 1)} \left(\text{$\sum_{e \in R}{f'(e)} \geq 1, C^* + 1 \le \frac{\Delta(H_Q)}{2}$ [Eq. \ref{eq:C_star}]}\right).
\end{split}
\end{equation*}
%\vspace{-0.05in}
\end{proof}
We also show a construction for the tight lower bound in Thm.~\ref{thm:local_pseudo_submodular} (Sec.~\ref{subsec:tight_thm2}).
%\vspace{-0.25in}

\subsection{Proof with bound $\frac{4\psi^r}{4\psi^r + \Delta^{RG} (b - 1)}$ } 
\label{subsec:tigh_tbound_rg}

In proof of Thm.~\ref{thm:local_pseudo_submodular}, we have $\sum_{e \in R}{f'(e)} \geq 1$. However, $\sum_{e \in R}{f'(e)} \geq \psi^r$, where $\psi^r$ is the summation of marginal gains of the elements in the optimal solution set (i.e., $R$) over the solution set produced by \rg (i.e., $Q$). Now replacing $\sum_{e \in R}{f'(e)}$, as $\psi^r$ we get, $\gamma' \geq \frac{4\psi^r}{4\psi^r + \Delta^{RG} (b - 1)} $ according to Thm. \ref{thm:approx_RG}.

\vspace{-0.05in}
\subsection{Approximation by \greedy}
\label{sec:greed_approx}

%\vspace{-0.05in}
%\subsubsection{Approximation via \greedy Algorithm}
%\label{subsubsec:appx_greedy}
%Besides the \textbf{RG} algorithm, we show that we achieve the same approximation guarantees via the simple \greedy algorithm as well. 

%For this we use the greedy algorithm (Algorithm 1) presented in \cite{Das-JMLR:2018}.

    % Consider $\frac{\sum_{e \in R}{[f(Q \cup \{e\})-f(Q)]}}{f(Q \cup R)- f(Q)}$ for sets $Q, R \in E$. Using the Lemma \ref{lemma:induction}, we get $\frac{\sum_{e \in R}{f'(\{e\})}}{f'(R)} = \frac{1}{1 + \frac{\alpha(R)(C^* + 1)}{\sum_{e \in R}{f'(\{e\})}}}$. 
    
    % Note that $\alpha(R) \le \frac{|R|-1}{2}$ (Obs. \ref{obs:hsize}), $\sum_{e \in R}{f'(e)} \geq 1, C^* \le \frac{\Delta(H_Q) - 2}{2}$. Substituting these, we get the desired result. 

%Assuming $0 \le \gamma_i \le 1$ for $i \in \{0, 1, 2, \cdots, k-1\}$ so that $\sum_{e \in OPT} {[f(S_i \cup\{e\})-f(S_i))]} \geq \gamma_i \cdot [f(S_i \cup OPT)- f(S_i)]$ (local pseduo-submodularity) throughout the execution of the \textbf{RG} algorithm, where $f$ is monotonic, $OPT$ denotes the optimal set of edges, and $S_i$ denotes the set of chosen elements after the $i$-th iteration (i.e. $|S_i| = i$); then \textbf{RG} obtains an approximation of $1 - \exp{\left(-\frac{1}{k}\sum_{i=0}^{k-1}{\gamma_i}\right)}$.

\begin{lem}[\cite{Das-JMLR:2018}]
\label{lem:Das_greedy}
%Given $f$ is a non-negative, monotone set function, and $OPT$ be the maximum value of $f$ obtained by any set of size $b$. $\sum_{e \in OPT} {[f(S_i \cup\{e\})-f(S_i))]} \geq \gamma_i \cdot [f(S_i \cup OPT)- f(S_i)]$Then, the final set $S^G$ selected by the \greedy Algorithm has the following approximation guarantee: $f(S^G)\geq (1-e^{-\gamma_{S^G,b}})\cdot OPT $.
Given $f$ is a non-negative and monotone set function, budget $b$, and $\sum_{e \in R} {[f(S^G \cup\{e\})-f(S^G))]} \geq \gamma  \cdot [f(S^G \cup R)- f(R)]$ where $S^G$ is the final set selected by the \greedy Algorithm, then the algorithm has the following approximation guarantee of $(1-e^{-\gamma_{S^G,b}})$ where $\gamma_{S^G,b} = min \{\gamma\}$ for any $R, S^G \cap R =\Phi$.
\end{lem}
\vspace{-0.05in}

We apply this result in our problem setting: 

\vspace{-0.05in}
\begin{thm} 
\label{thm:approx_greedy}
For the MBED problem, \greedy algorithm obtains an approximation of $1 - e^{-\gamma'}$, and $\gamma' \geq \frac{4}{4 + \Delta^{*}(b -1)}$
where $b$ and $\Delta^{*}$ denote the budget and the balance after deleting the optimal set of edges respectively.
\end{thm}
\vspace{-0.1in}

\begin{proof}
Let the optimal set of $b$ edges be $B^*$ and let $S^G$ denote the final edge set by the \greedy algorithm. Also, let $\Delta^{*}$ denote the balance after deleting the optimal set of edges, then by its definition we have $\Delta(H_{S^G}) \le \Delta^{*}$. From Theorem \ref{thm:local_pseudo_submodular}, $\gamma'= \gamma_{S^G, |B^*|} \geq \frac{4}{4 + \Delta(H_{S^G})(|B^*| - 1)} \geq \frac{4}{4 + \Delta^{*} (b -1)} $. So, substituting $\gamma_{S^G,b}$ in Lem. \ref{lem:Das_greedy} as $\gamma_{S^G, |B^*|}$ (or $\gamma'$), we get the desired approximation of $\frac{4}{4 + \Delta^{*} (b -1)}$. 
\end{proof}
%Let $\Delta^G=\Delta(H_{S^G})$ be the balance after deleting the edge set returned by the \greedy algorithm. We can obtain a tighter approximation bound than  Thm.~\ref{thm:approx_greedy}: $\gamma' \geq \frac{4}{4 + \Delta^{G} (b - 1)} $.

The other lowers bounds for $\gamma'$ (where the approximation produced by \greedy is $1 - e^{-\gamma'}$) as $\frac{4}{4 + \Delta^{G} (b - 1)}$ and $\frac{4\psi^g}{4\psi^g + \Delta^{G} (b - 1)}$ can be derived in similar ways as in the case of \rg.

%% file: appendix.tex
\section{Additional proofs}

\subsection{NP-hardness}
\label{app:nphard}
\begin{proof}

Let $SK(U,S,P,W,q)$ be an instance of the Set Union Knapsack Problem 
 \cite{goldschmidt1994note}, where $U=\{u_1, \ldots u_{n}\}$ is a set of items, $S=\{S_1, \ldots S_{m}\}$ is a set of subsets ($S_i \subseteq U$), $P:S\to\mathbb{R}_{+}$ is a subset profit function, $w:U\to\mathbb{R}_+$ is an item weight function, and $q \in \mathbb{R}_+$ is the budget. For a subset $\mathcal{A} \subseteq S$, the weighted union of set $\mathcal{A}$ is $W(\mathcal{A}) = \sum_{e\in \cup_{t\in \mathcal{A}} S_t} w_e$ and $P(\mathcal{A})=\sum_{t \in \mathcal{A}} p_t$. The problem is to find a subset $\mathcal{A}^*\subseteq S$ such that $W(\mathcal{A}^*)\leq q$ and $P(\mathcal{A}^*)$ is maximized.  SK is NP-hard to approximate within a constant factor \cite{arulselvan2014note}. We reduce a version of $SK$ with equal profits and weights (also NP-hard) to the \textsc{Mbed} problem.
 We define a corresponding \textsc{Mbed} problem instance via constructing a graph $\Gamma$ as follows. 

For each $S_i \in S$ and $u_j \in U$ we create nodes $x_i$ and $y_j$ respectively. We also add a node $v$ with a large connected component $L$ of size $l$ only with positive edges attached to it. The node $v$ has negative edges with every node $x_i$, $\forall i\in [m]$ and every node $y_j$, $\forall j \in [n]$. Additionally, if $u_j \in S_i$, a negative edge $(x_i,y_j)$ will be added to the edge set $E$. 

%We assume $k=q$ and $H=\Gamma$. Note that initial largest connected balanced component is $\{v \cup L\}\cup \{y_j \forall j \in [n]\}$ if $l>m+1$. We claim that a set $S'\subset S$, with $|S'|\leq r$, is a cover iff $f(B) = l+1+k+n$ where the solution set of $r=k$ edges, $B\!=\!\{(v,x_i)| S_i \in S' \}$.  

In \textsc{Mbed}, the number of edges to be removed is the budget, $b=q$. The candidate set, $\mathbb{C} = \{ (v,y_j)| \forall j\in [n]\}$. Note that initial largest connected balanced component is $\{v \cup L\}\cup \{y_j \forall j \in [n]\}$ if $l>m+1$ (assuming $n>m$). Our claim is that, for any solution $\mathcal{A}$ of an instance of $SK$ there is a corresponding solution set of edges, $B$ (where $|B|=b$) in the graph $\Gamma$ of the \textsc{Mbed} version, such that $f(B)=P(\mathcal{A})+n+l+1$ if  $B = \{(v,y)|y\in \mathcal{A}\}$ are removed.

In the new balanced graph, we aim to build two partitions ($W_1$ and $W_2$) as follows. One partition $W_1$ consists of $\{v \cup L\}$ initially. Our goal is to delete edges from $\mathbb{C}$ and add the nodes $y_j$'s in $W_1$. If $(v,y_{j'})$ for any $j'$ does not get deleted then it would be in $W_2$. If there is any node $x_i$ that is connected with only nodes in $\mathcal{A}$ beside being connected with $v$, then removing all the edges in $B$ would put the node $x_i$ in $W_2$. Thus removing edges in $\mathcal{A}$ would put $P(\mathcal{A})$ nodes in $W_2$. Thus, $f(B)=P(\mathcal{A})+n+l+1$.

\end{proof}
\subsection{Proportionally Submodular}
\label{app:propsubmodular}
\begin{lemma}\label{lemma:prop_submodular}
The objective function $f$ is not proportionally\\ submodular \cite{santiago2020weakly}. In other words, there exists $S, T \in E$ for some graph $H$ such that $|T|f(S) + |S|f(T) < |S \cap T| f(S \cup T) + |S \cup T| f(S \cap T)$.
\end{lemma}

\begin{proof}
    Consider a balanced subgraph of $H$, $S(H)$ has a partition $V_1$ and $V_2$. A node $v$ is outside $S(H)$ and it is connected to $V_1$ with positive edges $e_1$ and $e_2$, $V_2$ with another positive edge $e_3$. Thus the node $v$ cannot be the part of $S(H)$. Consider an edge $e_4$ inside $V_1$ which can be removed without making the graph disconnected. Let us assume $S = \{e_1, e_4\}, T = \{e_2, e_4\}$. Then, $f(\{e_1, e_4\}) = 0$ and $f(\{e_2, e_4\}) = 0$, since even after removing any of these edges it is not possible to add the node $v$ to $S(H)$. Note that $f(S \cap T) = f(\{e_4\}) = 0$. However, $f(S \cup T) = f(\{e_1, e_2, e_4\}) = 1$ since the node $v$ can be added. Substituting these values, we get $|T|f(S) + |S|f(T) < |S \cap T| f(S \cup T) + |S \cup T| f(S \cap T)$.
\end{proof}

\subsection{Proof of Lemma \ref{lemma:modular}}
\label{app:modular}

We denote $g_X(Y)$ as the marginal gain of the set of edges $Y$ over the set $X$, i.e., $g_X(Y) = g(X\cup Y)-g(X)$. To prove modularity, we need to show $g_X(Y)=\sum_{e \in Y} {g_X(e)}$, i.e. the marginal gain of the set of $Y$ over $X$ is the summation of the marginal gains of each individual in $Y$ over $X$ for any $X,Y$.

\vspace{-0.05in}
\begin{proof} We can write $g_X(Y)$ as follows.

\vspace{-0.05in}
    \begin{align*}
        g_X(Y) &= - \sum_{(i, j) \in X \cup Y} {{(\bm{v}_i - \sigma(i, j)\bm{v}_j)}^2} + \sum_{(i, j) \in X} {{(\bm{v}_i - \sigma(i, j)\bm{v}_j)}^2} \\
        &= -\sum_{(i, j) \in Y} {{(\bm{v}_i - \sigma(i, j)\bm{v}_j)}^2} 
        = \sum_{e \in Y} {g_X(e)}
    \end{align*}
\end{proof}
\begin{comment}
\subsection{Proof of Obs. \ref{obs:singleEdgeDeletion}}
\begin{proof}
    If $f(e) > 0$ then $e$ must be outside $S(H)$. Thus, $e$ must be connected to a node of any type noted in Observation~\ref{obs:nodetypes}. The balance cannot be increased if it is connected to a node of type $(2)$. Thus, $|cep(H, x)| > 0$ for $x$ being an endpoint of $e$. One cannot bring a node into $S(H)$ (and increase the balance) unless $|cep(H, x)| = 0$. Thus, the claim is proved. The other side of the proof directly follows from the definitions. 
\end{proof}
\end{comment}
%\subsection{Proof of Obs. \ref{obs:c_star_UB}}
%\label{sec:c_star_UB}

%\begin{proof}

%\vspace{-0.05in} 

\subsection{Proof of Lemma \ref{lemma:induction}}
\label{sec:proof_lemmma_induction}

\begin{proof}
    We prove this by induction on the number of edges, $b$. Let us denote $B_k\subseteq B$ as $\{e_1, \cdots, e_{k}\}$. We construct $B$ by only considering peripheral edges $e_{k+1}$ such that, for all $k \le b$: $(e_{k+1}, e') \in cep (H_{B_k}, x)$, for some node $x$ and edge $e'$.
    
    \textbf{Base case $(b =1)$:} $f(\{e_1\}) \le f(\{e_1\})$. Also, $\alpha(\{e_1\}) = 0$. 
    
    \textbf{Inductive hypothesis (IH):} Suppose the equation holds for $b = k$, i.e., $f(B_k) \le \sum_{i=1}^{k}{f(\{e_i\})} + (C^* + 1) \alpha(B_k)$.
    
    \textbf{Inductive step $(b = k+1)$:}
     We present different cases for $e_{k+1}$. Note that we have $(e_{k+1}, e') \in cep(H_{B_k}, x)$ for some $x, e'$. 
     
    \textbf{Case 1:} $(e_{k+1}, e') \in cep(H, x)$ and $\left|cep\left(H_{e_{k+1}}, x\right)\right| = 0$, i.e., after deleting $e_{k+1}$, $x$ moves into the balanced subgraph. Then, we must also have $\left|cep\left(H_{B_{k+1}}, x\right)\right| = 0$. Hence, $f(B_k \cup \{e_{k+1}\}) - f(B_k) = f(\{e_{k+1}\})$ and the inequality holds.

    \textbf{Case 2:} Either (1) $(e_{k+1}, e') \in cep(H, x)$ and $\left|cep\left(H_{e_{k+1}}, x\right)\right| > 0$ or (2) $(e_{k+1}, e') \notin cep(H, x)$. 
    
    Thus, by Observation \ref{obs:nodetypes}, we have $f(\{e_{k+1}\}) = 0$. 
    
    \textbf{Case 2a:} Suppose $\left|cep\left(H_{B_{k+1}}, x\right)\right| = 0$. Then by definition of $\alpha, C^*$,  we have  $\alpha(B_{k+1}) = \alpha(B_k) + 1$, and $f(\{B_k \cup e_{k+1}\}) - f(B_k) \le C^* + 1$. 
    
    Substituting this, we get $f(B_{k+1}) \le (C^* + 1) + \sum_{i=1}^{k}{f(\{e_i\})} +(C^* + 1) \alpha(B_k) = \sum_{i=1}^{k+1}{f(\{e_i\})} +(C^* + 1) \alpha(B_{k+1})$. 
    
    \textbf{Case 2b:} In other cases, $f(B_k \cup \{e_{k+1}\}) - f(B_k) = f(\{e_{k+1}\}) = 0$. 
    
    This exhausts our cases and the claim is true $\forall b, b > 0$. 
    % On the other hand, if $(e_{k+1}, e') \notin cep\left(H_{B_k}, x\right)$ or $\left|cep\left(H_{B_{k+1}}, x\right)\right| > 0$ for every edge $e'$, node $x$, then $f_{B_k}(B_{k+1}) = f(e_{k+1}) = 0$. This exhausts our cases and thus, by induction, it is true $\forall k, k > 0$. 
\end{proof}

\subsection{Construction for the tight lower bound in Thm. \ref{thm:local_pseudo_submodular}}
\label{subsec:tight_thm2}

One can construct a graph $H$ and the sets $Q, R$ where equality holds. In particular, let $R$ be of an arbitrary size $b$. Consider $H_Q$ to have the MBS partition as $V_1, V_2$ each of size $\frac{\Delta(H_Q)}{2}$. Nodes of type 1 (Obs. \ref{obs:nodetypes}) are attached to these each with the sole connected component of size $\frac{\Delta(H_Q) - 2}{2}$. Let these nodes have 3 such connections (thus, removing two will help - any two such that our "connected assumption" holds are in the set $R$). We have another node of type 1 such that only two such connections are connected and one of these is in $R$ and the connected component $C$ to it is of size $0$. This completes the set $R$. Thus, $\sum_{e \in R}{[f(Q \cup \{e\})- f(Q)} = 1$ and $f(Q\cup R)-f(Q) = 1 + \left(\frac{\Delta(H_Q) - 2}{2} + 1\right) \frac{b-1}{2}$.

%% file: WSDM21.bbl

\begin{thebibliography}{39}

%%% the Web address.

\ifx \showCODEN    \undefined \def \showCODEN     #1{\unskip}     \fi
\ifx \showDOI      \undefined \def \showDOI       #1{#1}\fi
\ifx \showISBNx    \undefined \def \showISBNx     #1{\unskip}     \fi
\ifx \showISBNxiii \undefined \def \showISBNxiii  #1{\unskip}     \fi
\ifx \showISSN     \undefined \def \showISSN      #1{\unskip}     \fi
\ifx \showLCCN     \undefined \def \showLCCN      #1{\unskip}     \fi
\ifx \shownote     \undefined \def \shownote      #1{#1}          \fi
\ifx \showarticletitle \undefined \def \showarticletitle #1{#1}   \fi
\ifx \showURL      \undefined \def \showURL       {\relax}        \fi
% invisible to TeX
\providecommand\bibfield[2]{#2}
\providecommand\bibinfo[2]{#2}
\providecommand\natexlab[1]{#1}
\providecommand\showeprint[2][]{arXiv:#2}

\bibitem[\protect\citeauthoryear{Akiyama, Avis, Chv{\'a}tal, and Era}{Akiyama
  et~al\mbox{.}}{1981}]%
        {akiyama1981balancing}
\bibfield{author}{\bibinfo{person}{Jin Akiyama}, \bibinfo{person}{David Avis},
  \bibinfo{person}{Vasek Chv{\'a}tal}, {and} \bibinfo{person}{Hiroshi Era}.}
  \bibinfo{year}{1981}\natexlab{}.
\newblock \showarticletitle{Balancing signed graphs}.
\newblock \bibinfo{journal}{\emph{Discrete Applied Mathematics}}
  \bibinfo{volume}{3}, \bibinfo{number}{4} (\bibinfo{year}{1981}),
  \bibinfo{pages}{227--233}.
\newblock


\bibitem[\protect\citeauthoryear{Arulselvan}{Arulselvan}{2014}]%
        {arulselvan2014note}
\bibfield{author}{\bibinfo{person}{Ashwin Arulselvan}.}
  \bibinfo{year}{2014}\natexlab{}.
\newblock \showarticletitle{A note on the set union knapsack problem}.
\newblock \bibinfo{journal}{\emph{Discrete Applied Mathematics}}
  \bibinfo{volume}{169} (\bibinfo{year}{2014}), \bibinfo{pages}{214--218}.
\newblock


\bibitem[\protect\citeauthoryear{Askarisichani, {Ng~Lane}, Bullo, Friedkin,
  Singh, and Uzzi}{Askarisichani et~al\mbox{.}}{2019}]%
        {omid}
\bibfield{author}{\bibinfo{person}{O. Askarisichani}, \bibinfo{person}{J.
  {Ng~Lane}}, \bibinfo{person}{F. Bullo}, \bibinfo{person}{N.~E. Friedkin},
  \bibinfo{person}{A.~K. Singh}, {and} \bibinfo{person}{B. Uzzi}.}
  \bibinfo{year}{2019}\natexlab{}.
\newblock \showarticletitle{Structural Balance Emerges and Explains Performance
  in Risky Decision-Making}.
\newblock  \bibinfo{volume}{10}, \bibinfo{number}{2648} (\bibinfo{year}{2019}).
\newblock
\urldef\tempurl%
\url{https://doi.org/10.1038/s41467-019-10548-8}
\showDOI{\tempurl}


\bibitem[\protect\citeauthoryear{Belardo}{Belardo}{2014}]%
        {belardo2014balancedness}
\bibfield{author}{\bibinfo{person}{Francesco Belardo}.}
  \bibinfo{year}{2014}\natexlab{}.
\newblock \showarticletitle{Balancedness and the least eigenvalue of Laplacian
  of signed graphs}.
\newblock \bibinfo{journal}{\emph{Linear Algebra Appl.}}  \bibinfo{volume}{446}
  (\bibinfo{year}{2014}), \bibinfo{pages}{133--147}.
\newblock


\bibitem[\protect\citeauthoryear{Chaoji, Ranu, Rastogi, and Bhatt}{Chaoji
  et~al\mbox{.}}{2012}]%
        {chaoji2012recommendations}
\bibfield{author}{\bibinfo{person}{Vineet Chaoji}, \bibinfo{person}{Sayan
  Ranu}, \bibinfo{person}{Rajeev Rastogi}, {and} \bibinfo{person}{Rushi
  Bhatt}.} \bibinfo{year}{2012}\natexlab{}.
\newblock \showarticletitle{Recommendations to boost content spread in social
  networks}. In \bibinfo{booktitle}{\emph{WWW}}. \bibinfo{pages}{529--538}.
\newblock


\bibitem[\protect\citeauthoryear{Coakley and Rokhlin}{Coakley and
  Rokhlin}{2013}]%
        {coakley2013fast}
\bibfield{author}{\bibinfo{person}{Ed~S Coakley} {and}
  \bibinfo{person}{Vladimir Rokhlin}.} \bibinfo{year}{2013}\natexlab{}.
\newblock \showarticletitle{A fast divide-and-conquer algorithm for computing
  the spectra of real symmetric tridiagonal matrices}.
\newblock \bibinfo{journal}{\emph{Applied and Computational Harmonic Analysis}}
  \bibinfo{volume}{34}, \bibinfo{number}{3} (\bibinfo{year}{2013}),
  \bibinfo{pages}{379--414}.
\newblock


\bibitem[\protect\citeauthoryear{Crescenzi, D'Angelo, Severini, and
  Velaj}{Crescenzi et~al\mbox{.}}{2015}]%
        {crescenzi2015}
\bibfield{author}{\bibinfo{person}{Pierluigi Crescenzi},
  \bibinfo{person}{Gianlorenzo D'Angelo}, \bibinfo{person}{Lorenzo Severini},
  {and} \bibinfo{person}{Yllka Velaj}.} \bibinfo{year}{2015}\natexlab{}.
\newblock \showarticletitle{Greedily Improving Our Own Centrality in A
  Network}. In \bibinfo{booktitle}{\emph{SEA}}. \bibinfo{publisher}{Springer
  International Publishing}, \bibinfo{pages}{43--55}.
\newblock


\bibitem[\protect\citeauthoryear{Das and Kempe}{Das and Kempe}{2018}]%
        {Das-JMLR:2018}
\bibfield{author}{\bibinfo{person}{A. Das} {and} \bibinfo{person}{D. Kempe}.}
  \bibinfo{year}{2018}\natexlab{}.
\newblock \showarticletitle{Approximate submodularity and its applications:
  subset selection, sparse approximation and dictionary selection}.
\newblock \bibinfo{journal}{\emph{The Journal of Machine Learning Research}}
  \bibinfo{volume}{19}, \bibinfo{number}{1} (\bibinfo{year}{2018}),
  \bibinfo{pages}{74--107}.
\newblock


\bibitem[\protect\citeauthoryear{DasGupta, Enciso, Sontag, and Zhang}{DasGupta
  et~al\mbox{.}}{2007}]%
        {dasgupta2007algorithmic}
\bibfield{author}{\bibinfo{person}{Bhaskar DasGupta},
  \bibinfo{person}{German~Andres Enciso}, \bibinfo{person}{Eduardo Sontag},
  {and} \bibinfo{person}{Yi Zhang}.} \bibinfo{year}{2007}\natexlab{}.
\newblock \showarticletitle{Algorithmic and complexity results for
  decompositions of biological networks into monotone subsystems}.
\newblock \bibinfo{journal}{\emph{Biosystems}} \bibinfo{volume}{90},
  \bibinfo{number}{1} (\bibinfo{year}{2007}), \bibinfo{pages}{161--178}.
\newblock


\bibitem[\protect\citeauthoryear{Dey and Medya}{Dey and Medya}{2020}]%
        {dey2019manipulating}
\bibfield{author}{\bibinfo{person}{Palash Dey} {and} \bibinfo{person}{Sourav
  Medya}.} \bibinfo{year}{2020}\natexlab{}.
\newblock \showarticletitle{Manipulating Node Similarity Measures in Network}.
  In \bibinfo{booktitle}{\emph{AAMAS}}.
\newblock


\bibitem[\protect\citeauthoryear{Dilkina, Lai, and Gomes}{Dilkina
  et~al\mbox{.}}{2011}]%
        {dilkina2011}
\bibfield{author}{\bibinfo{person}{Bistra Dilkina},
  \bibinfo{person}{Katherine~J. Lai}, {and} \bibinfo{person}{Carla~P. Gomes}.}
  \bibinfo{year}{2011}\natexlab{}.
\newblock \showarticletitle{Upgrading shortest paths in networks}. In
  \bibinfo{booktitle}{\emph{Integration of AI and OR Techniques in Constraint
  Programming for Combinatorial Optimization Problems}}.
  \bibinfo{publisher}{Springer}, \bibinfo{pages}{76--91}.
\newblock


\bibitem[\protect\citeauthoryear{Figueiredo and Frota}{Figueiredo and
  Frota}{2014}]%
        {figueiredo2014maximum}
\bibfield{author}{\bibinfo{person}{Rosa Figueiredo} {and} \bibinfo{person}{Yuri
  Frota}.} \bibinfo{year}{2014}\natexlab{}.
\newblock \showarticletitle{The maximum balanced subgraph of a signed graph:
  Applications and solution approaches}.
\newblock \bibinfo{journal}{\emph{European Journal of Operational Research}}
  \bibinfo{volume}{236}, \bibinfo{number}{2} (\bibinfo{year}{2014}),
  \bibinfo{pages}{473--487}.
\newblock


\bibitem[\protect\citeauthoryear{Garimella and Weber}{Garimella and
  Weber}{2017}]%
        {garimella2017long}
\bibfield{author}{\bibinfo{person}{Venkata Rama~Kiran Garimella} {and}
  \bibinfo{person}{Ingmar Weber}.} \bibinfo{year}{2017}\natexlab{}.
\newblock \showarticletitle{A long-term analysis of polarization on Twitter}.
  In \bibinfo{booktitle}{\emph{Eleventh International AAAI Conference on Web
  and Social Media}}.
\newblock


\bibitem[\protect\citeauthoryear{Goldschmidt, Nehme, and Yu}{Goldschmidt
  et~al\mbox{.}}{1994}]%
        {goldschmidt1994note}
\bibfield{author}{\bibinfo{person}{Olivier Goldschmidt}, \bibinfo{person}{David
  Nehme}, {and} \bibinfo{person}{Gang Yu}.} \bibinfo{year}{1994}\natexlab{}.
\newblock \showarticletitle{Note: On the set-union knapsack problem}.
\newblock \bibinfo{journal}{\emph{Naval Research Logistics (NRL)}}
  (\bibinfo{year}{1994}).
\newblock


\bibitem[\protect\citeauthoryear{Harary et~al\mbox{.}}{Harary
  et~al\mbox{.}}{1953}]%
        {harary1953notion}
\bibfield{author}{\bibinfo{person}{Frank Harary} {et~al\mbox{.}}}
  \bibinfo{year}{1953}\natexlab{}.
\newblock \showarticletitle{On the notion of balance of a signed graph.}
\newblock \bibinfo{journal}{\emph{The Michigan Mathematical Journal}}
  \bibinfo{volume}{2}, \bibinfo{number}{2} (\bibinfo{year}{1953}),
  \bibinfo{pages}{143--146}.
\newblock


\bibitem[\protect\citeauthoryear{Hou, Li, and Pan}{Hou et~al\mbox{.}}{2003}]%
        {hou2003laplacian}
\bibfield{author}{\bibinfo{person}{Yaoping Hou}, \bibinfo{person}{Jiongsheng
  Li}, {and} \bibinfo{person}{Yongliang Pan}.} \bibinfo{year}{2003}\natexlab{}.
\newblock \showarticletitle{On the Laplacian eigenvalues of signed graphs}.
\newblock \bibinfo{journal}{\emph{Linear and Multilinear Algebra}}
  \bibinfo{volume}{51}, \bibinfo{number}{1} (\bibinfo{year}{2003}),
  \bibinfo{pages}{21--30}.
\newblock


\bibitem[\protect\citeauthoryear{H{\"u}ffner, Betzler, and
  Niedermeier}{H{\"u}ffner et~al\mbox{.}}{2007}]%
        {huffner2007optimal}
\bibfield{author}{\bibinfo{person}{Falk H{\"u}ffner}, \bibinfo{person}{Nadja
  Betzler}, {and} \bibinfo{person}{Rolf Niedermeier}.}
  \bibinfo{year}{2007}\natexlab{}.
\newblock \showarticletitle{Optimal edge deletions for signed graph balancing}.
  In \bibinfo{booktitle}{\emph{International Workshop on Experimental and
  Efficient Algorithms}}. Springer, \bibinfo{pages}{297--310}.
\newblock


\bibitem[\protect\citeauthoryear{Ishakian, Erdos, Terzi, and
  Bestavros}{Ishakian et~al\mbox{.}}{2012}]%
        {ishakian2012framework}
\bibfield{author}{\bibinfo{person}{Vatche Ishakian}, \bibinfo{person}{D{\'o}ra
  Erdos}, \bibinfo{person}{Evimaria Terzi}, {and} \bibinfo{person}{Azer
  Bestavros}.} \bibinfo{year}{2012}\natexlab{}.
\newblock \showarticletitle{A Framework for the Evaluation and Management of
  Network Centrality}. In \bibinfo{booktitle}{\emph{Proc. SIAM International
  Conference on Data Mining}}. \bibinfo{pages}{427--438}.
\newblock


\bibitem[\protect\citeauthoryear{Kempe, Kleinberg, and Tardos}{Kempe
  et~al\mbox{.}}{2003}]%
        {kempe}
\bibfield{author}{\bibinfo{person}{David Kempe}, \bibinfo{person}{Jon
  Kleinberg}, {and} \bibinfo{person}{{\'E}va Tardos}.}
  \bibinfo{year}{2003}\natexlab{}.
\newblock \showarticletitle{Maximizing the spread of influence through a social
  network}. In \bibinfo{booktitle}{\emph{KDD}}.
\newblock


\bibitem[\protect\citeauthoryear{Kimura, Saito, and Motoda}{Kimura
  et~al\mbox{.}}{2008}]%
        {kimura2008minimizing}
\bibfield{author}{\bibinfo{person}{Masahiro Kimura}, \bibinfo{person}{Kazumi
  Saito}, {and} \bibinfo{person}{Hiroshi Motoda}.}
  \bibinfo{year}{2008}\natexlab{}.
\newblock \showarticletitle{Minimizing the Spread of Contamination by Blocking
  Links in a Network.}. In \bibinfo{booktitle}{\emph{AAAI}}.
\newblock


\bibitem[\protect\citeauthoryear{Knyazev}{Knyazev}{2001}]%
        {knyazev2001toward}
\bibfield{author}{\bibinfo{person}{Andrew~V Knyazev}.}
  \bibinfo{year}{2001}\natexlab{}.
\newblock \showarticletitle{Toward the optimal preconditioned eigensolver:
  Locally optimal block preconditioned conjugate gradient method}.
\newblock \bibinfo{journal}{\emph{SIAM journal on scientific computing}}
  \bibinfo{volume}{23}, \bibinfo{number}{2} (\bibinfo{year}{2001}),
  \bibinfo{pages}{517--541}.
\newblock


\bibitem[\protect\citeauthoryear{Li and Li}{Li and Li}{2009}]%
        {li2009note}
\bibfield{author}{\bibinfo{person}{Hong-hai Li} {and}
  \bibinfo{person}{Jiong-sheng Li}.} \bibinfo{year}{2009}\natexlab{}.
\newblock \showarticletitle{Note on the normalized Laplacian eigenvalues of
  signed graphs.}
\newblock \bibinfo{journal}{\emph{Australasian J. Combinatorics}}
  \bibinfo{volume}{44} (\bibinfo{year}{2009}), \bibinfo{pages}{153--162}.
\newblock


\bibitem[\protect\citeauthoryear{Lin and Mouratidis}{Lin and
  Mouratidis}{2015}]%
        {lin2015}
\bibfield{author}{\bibinfo{person}{Yimin Lin} {and} \bibinfo{person}{Kyriakos
  Mouratidis}.} \bibinfo{year}{2015}\natexlab{}.
\newblock \showarticletitle{Best upgrade plans for single and multiple
  source-destination pairs}.
\newblock \bibinfo{journal}{\emph{GeoInformatica}} \bibinfo{volume}{19},
  \bibinfo{number}{2} (\bibinfo{year}{2015}), \bibinfo{pages}{365--404}.
\newblock


\bibitem[\protect\citeauthoryear{Medya, Ma, Silva, and Singh}{Medya
  et~al\mbox{.}}{2020}]%
        {medyakcore}
\bibfield{author}{\bibinfo{person}{Sourav Medya}, \bibinfo{person}{Tiyani Ma},
  \bibinfo{person}{Arlei Silva}, {and} \bibinfo{person}{Ambuj Singh}.}
  \bibinfo{year}{2020}\natexlab{}.
\newblock \showarticletitle{A Game Theoretic Approach For Core Resilience}. In
  \bibinfo{booktitle}{\emph{Proceedings of the Twenty-Ninth International Joint
  Conference on Artificial Intelligence, {IJCAI-20}}}.
\newblock


\bibitem[\protect\citeauthoryear{{Medya}, {Silva}, and {Singh}}{{Medya}
  et~al\mbox{.}}{2020}]%
        {medya_influence}
\bibfield{author}{\bibinfo{person}{S. {Medya}}, \bibinfo{person}{A. {Silva}},
  {and} \bibinfo{person}{A. {Singh}}.} \bibinfo{year}{2020}\natexlab{}.
\newblock \showarticletitle{Approximate Algorithms for Data-driven Influence
  Limitation}.
\newblock \bibinfo{journal}{\emph{IEEE Transactions on Knowledge and Data
  Engineering}} (\bibinfo{year}{2020}).
\newblock


\bibitem[\protect\citeauthoryear{Medya, Silva, Singh, Basu, and Swami}{Medya
  et~al\mbox{.}}{2018a}]%
        {medya2018group}
\bibfield{author}{\bibinfo{person}{Sourav Medya}, \bibinfo{person}{Arlei
  Silva}, \bibinfo{person}{Ambuj Singh}, \bibinfo{person}{Prithwish Basu},
  {and} \bibinfo{person}{Ananthram Swami}.} \bibinfo{year}{2018}\natexlab{a}.
\newblock \showarticletitle{Group centrality maximization via network design}.
  In \bibinfo{booktitle}{\emph{Proc. 24th SIAM International Conference on Data
  Mining}}. SIAM, \bibinfo{pages}{126--134}.
\newblock


\bibitem[\protect\citeauthoryear{Medya, Vachery, Ranu, and Singh}{Medya
  et~al\mbox{.}}{2018b}]%
        {medya2018noticeable}
\bibfield{author}{\bibinfo{person}{Sourav Medya}, \bibinfo{person}{Jithin
  Vachery}, \bibinfo{person}{Sayan Ranu}, {and} \bibinfo{person}{Ambuj Singh}.}
  \bibinfo{year}{2018}\natexlab{b}.
\newblock \showarticletitle{Noticeable network delay minimization via node
  upgrades}.
\newblock \bibinfo{journal}{\emph{Proceedings of the VLDB Endowment}}
  \bibinfo{volume}{11}, \bibinfo{number}{9} (\bibinfo{year}{2018}),
  \bibinfo{pages}{988--1001}.
\newblock


\bibitem[\protect\citeauthoryear{Meyerson and Tagiku}{Meyerson and
  Tagiku}{2009}]%
        {meyerson2009}
\bibfield{author}{\bibinfo{person}{Adam Meyerson} {and} \bibinfo{person}{Brian
  Tagiku}.} \bibinfo{year}{2009}\natexlab{}.
\newblock \showarticletitle{Minimizing average shortest path distances via
  shortcut edge addition}.
\newblock In \bibinfo{booktitle}{\emph{Approximation, Randomization, and
  Combinatorial Optimization. Algorithms and Techniques (APPROX-RANDOM)}}.
  \bibinfo{publisher}{Springer}, \bibinfo{pages}{272--285}.
\newblock


\bibitem[\protect\citeauthoryear{Mitra, Ranu, Kolar, Telang, Bhattacharya,
  Kokku, and Raghavan}{Mitra et~al\mbox{.}}{2015}]%
        {myinfocom}
\bibfield{author}{\bibinfo{person}{Shubhadip Mitra}, \bibinfo{person}{Sayan
  Ranu}, \bibinfo{person}{Vinay Kolar}, \bibinfo{person}{Aditya Telang},
  \bibinfo{person}{Arnab Bhattacharya}, \bibinfo{person}{Ravi Kokku}, {and}
  \bibinfo{person}{Sriram Raghavan}.} \bibinfo{year}{2015}\natexlab{}.
\newblock \showarticletitle{Trajectory aware macro-cell planning for mobile
  users}. In \bibinfo{booktitle}{\emph{2015 IEEE Conference on Computer
  Communications (INFOCOM)}}. IEEE, \bibinfo{pages}{792--800}.
\newblock


\bibitem[\protect\citeauthoryear{Nemhauser and Wolsey}{Nemhauser and
  Wolsey}{1978}]%
        {nemhauser1978}
\bibfield{author}{\bibinfo{person}{George~L Nemhauser} {and}
  \bibinfo{person}{Laurence~A Wolsey}.} \bibinfo{year}{1978}\natexlab{}.
\newblock \showarticletitle{Best algorithms for approximating the maximum of a
  submodular set function}.
\newblock \bibinfo{journal}{\emph{Mathematics of operations research}}
  \bibinfo{volume}{3}, \bibinfo{number}{3} (\bibinfo{year}{1978}),
  \bibinfo{pages}{177--188}.
\newblock


\bibitem[\protect\citeauthoryear{Ordozgoiti, Matakos, and Gionis}{Ordozgoiti
  et~al\mbox{.}}{2020}]%
        {ordozgoiti2020finding}
\bibfield{author}{\bibinfo{person}{Bruno Ordozgoiti}, \bibinfo{person}{Antonis
  Matakos}, {and} \bibinfo{person}{Aristides Gionis}.}
  \bibinfo{year}{2020}\natexlab{}.
\newblock \showarticletitle{Finding large balanced subgraphs in signed
  networks}. In \bibinfo{booktitle}{\emph{Proceedings of The Web Conference
  2020}}. \bibinfo{pages}{1378--1388}.
\newblock


\bibitem[\protect\citeauthoryear{Orecchia, Sachdeva, and Vishnoi}{Orecchia
  et~al\mbox{.}}{2012}]%
        {orecchia2012approximating}
\bibfield{author}{\bibinfo{person}{Lorenzo Orecchia}, \bibinfo{person}{Sushant
  Sachdeva}, {and} \bibinfo{person}{Nisheeth~K Vishnoi}.}
  \bibinfo{year}{2012}\natexlab{}.
\newblock \showarticletitle{Approximating the exponential, the Lanczos method
  and an O (m)-time spectral algorithm for balanced separator}. In
  \bibinfo{booktitle}{\emph{Proceedings of the forty-fourth annual ACM
  symposium on Theory of computing}}. \bibinfo{pages}{1141--1160}.
\newblock


\bibitem[\protect\citeauthoryear{Paulheim}{Paulheim}{2017}]%
        {kg}
\bibfield{author}{\bibinfo{person}{Heiko Paulheim}.}
  \bibinfo{year}{2017}\natexlab{}.
\newblock \showarticletitle{Knowledge graph refinement: A survey of approaches
  and evaluation methods}.
\newblock \bibinfo{journal}{\emph{Semantic web}} \bibinfo{volume}{8},
  \bibinfo{number}{3} (\bibinfo{year}{2017}), \bibinfo{pages}{489--508}.
\newblock


\bibitem[\protect\citeauthoryear{Peng, Kolda, and Pinar}{Peng
  et~al\mbox{.}}{2014}]%
        {peng2014accelerating}
\bibfield{author}{\bibinfo{person}{Chengbin Peng}, \bibinfo{person}{Tamara~G
  Kolda}, {and} \bibinfo{person}{Ali Pinar}.} \bibinfo{year}{2014}\natexlab{}.
\newblock \showarticletitle{Accelerating community detection by using k-core
  subgraphs}.
\newblock \bibinfo{journal}{\emph{arXiv preprint arXiv:1403.2226}}
  (\bibinfo{year}{2014}).
\newblock


\bibitem[\protect\citeauthoryear{Poljak and Turz{\'\i}k}{Poljak and
  Turz{\'\i}k}{1986}]%
        {poljak1986polynomial}
\bibfield{author}{\bibinfo{person}{Svatopluk Poljak} {and}
  \bibinfo{person}{Daniel Turz{\'\i}k}.} \bibinfo{year}{1986}\natexlab{}.
\newblock \showarticletitle{A polynomial time heuristic for certain subgraph
  optimization problems with guaranteed worst case bound}.
\newblock \bibinfo{journal}{\emph{Discrete Mathematics}} \bibinfo{volume}{58},
  \bibinfo{number}{1} (\bibinfo{year}{1986}), \bibinfo{pages}{99--104}.
\newblock


\bibitem[\protect\citeauthoryear{Santiago and Yoshida}{Santiago and
  Yoshida}{2020}]%
        {santiago2020weakly}
\bibfield{author}{\bibinfo{person}{Richard Santiago} {and}
  \bibinfo{person}{Yuichi Yoshida}.} \bibinfo{year}{2020}\natexlab{}.
\newblock \showarticletitle{Weakly Submodular Function Maximization Using Local
  Submodularity Ratio}.
\newblock \bibinfo{journal}{\emph{arXiv preprint arXiv:2004.14650}}
  (\bibinfo{year}{2020}).
\newblock


\bibitem[\protect\citeauthoryear{Stewart and Sun}{Stewart and Sun}{1990}]%
        {Stewart-Book:1990}
\bibfield{author}{\bibinfo{person}{G.W. Stewart} {and} \bibinfo{person}{J-g
  Sun}.} \bibinfo{year}{1990}\natexlab{}.
\newblock \bibinfo{booktitle}{\emph{Matrix Perturbation Theory}}.
\newblock \bibinfo{publisher}{Academic Press, Inc.}
\newblock


\bibitem[\protect\citeauthoryear{Zhang, Zhang, Qin, Zhang, and Lin}{Zhang
  et~al\mbox{.}}{2017}]%
        {zhang2017finding}
\bibfield{author}{\bibinfo{person}{Fan Zhang}, \bibinfo{person}{Ying Zhang},
  \bibinfo{person}{Lu Qin}, \bibinfo{person}{Wenjie Zhang}, {and}
  \bibinfo{person}{Xuemin Lin}.} \bibinfo{year}{2017}\natexlab{}.
\newblock \showarticletitle{Finding Critical Users for Social Network
  Engagement: The Collapsed k-Core Problem}. In
  \bibinfo{booktitle}{\emph{Thirty-First AAAI Conference on Artificial
  Intelligence}}. \bibinfo{pages}{245--251}.
\newblock


\bibitem[\protect\citeauthoryear{Zhou, Zhang, Lin, Zhang, and Chen}{Zhou
  et~al\mbox{.}}{2019}]%
        {zhou2019k}
\bibfield{author}{\bibinfo{person}{Zhongxin Zhou}, \bibinfo{person}{Fan Zhang},
  \bibinfo{person}{Xuemin Lin}, \bibinfo{person}{Wenjie Zhang}, {and}
  \bibinfo{person}{Chen Chen}.} \bibinfo{year}{2019}\natexlab{}.
\newblock \showarticletitle{K-Core Maximization: An Edge Addition Approach.}.
  In \bibinfo{booktitle}{\emph{IJCAI}}. \bibinfo{pages}{4867--4873}.
\newblock


\end{thebibliography}
